\documentclass[11pt,twoside]{article}
\usepackage[margin=1in]{geometry}
\usepackage{jeffe}

\usepackage[mathletters]{ucs}		
\usepackage[utf8x]{inputenc}

\usepackage{microtype}
\usepackage{graphicx}
\DeclareGraphicsExtensions{.pdf,.png,.jpg}

\usepackage[charter]{mathdesign}	
\usepackage{berasans,beramono}
\SetMathAlphabet{\mathsf}{bold}{\encodingdefault}{\sfdefault}{b}{\updefault}
\SetMathAlphabet{\mathtt}{bold}{\encodingdefault}{\ttdefault}{b}{\updefault}
\SetMathAlphabet{\mathsf}{normal}{\encodingdefault}{\sfdefault}{\mddefault}{\updefault}
\SetMathAlphabet{\mathtt}{normal}{\encodingdefault}{\ttdefault}{\mddefault}{\updefault}
\RequirePackage{microtype}
\usepackage[mathcal]{eucal}

\usepackage{cite}

\usepackage{stmaryrd,textcomp}
\def\shortleftrightarrow{{\rlap{$\mathord\shortleftarrow$}\,\mathord\shortrightarrow}}
\def\arcto{\mathord\shortrightarrow}

\def\arc#1#2{#1\arcto#2}

\def\biarc#1#2{#1\mathord\shortleftrightarrow#2}

\def\etal{\emph{et~al.}}			
\def\Real{\mathbb{R}}
\let\e\varepsilon

\def\sgn{\operatorname{sgn}}
\def\floor#1{\lfloor #1 \rfloor}

\def\set#1{\{ #1 \}}
\def\abs#1{\mathopen| #1 \mathclose|}		

\def\Set#1{\left\{ #1 \right\}}
\def\Abs#1{\left| #1 \right|}

\def\Depth{\operatorname{\mathit{depth}}}
\def\Defect{\operatorname{\mathit{defect}}}
\def\Wind{\operatorname{\mathit{wind}}}

\usepackage{booktabs}
\usepackage{mathtools}	
\usepackage{wrapfig}
\usepackage{soul}		

\definecolor{Edit}{cmyk}{0,1,1,0}
\let\EDIT\relax

\newtheorem{theorem}{Theorem}[section]
\newtheorem{lemma}[theorem]{Lemma}
\newtheorem{corollary}[theorem]{Corollary}
\numberwithin{figure}{section}

\urldef\paperURL\url{http://jeffe.cs.illinois.edu/pubs/tangle.pdf}

\begin{document}

\begin{titlepage}

\title{Untangling Planar Curves%
\thanks{Work on this paper was partially supported by NSF grant CCF-1408763.  A preliminary version of this paper was presented at the 32nd International Symposium on Computational Geometry \cite{tangle}.  See \paperURL\ for the most recent version of this paper.}}

\author{Hsien-Chih Chang \qquad Jeff Erickson\\[1ex]
Department of Computer Science\\
University of Illinois, Urbana-Champaign\\
\href{mailto:hchang17@illinois.edu,jeffe@illinois.edu}{\{hchang17, jeffe\}@illinois.edu}}

%
%



\date{Submitted to \emph{Discrete \& Computational Geometry} --- July 16, 2016
\\[1ex]
Revised and resubmitted --- \today}

\maketitle

\begin{bigabstract}
Any generic closed curve in the plane can be transformed into a simple closed curve by a finite sequence of local transformations called \emph{homotopy moves}. We prove that simplifying a planar closed curve with $n$ self-crossings requires $\Theta(n^{3/2})$ homotopy moves in the worst case.  Our algorithm improves the best previous upper bound $O(n^2)$, which is already implicit in the classical work of Steinitz; the matching lower bound follows from the construction of closed curves with large \emph{defect}, a topological invariant of generic closed curves introduced by Aicardi and Arnold.  Our lower bound also implies that $\Omega(n^{3/2})$ \EDIT{facial electrical transformations} are required to reduce any \EDIT{plane} graph with treewidth $\Omega(\sqrt{n})$ to a single \EDIT{vertex}, matching known upper bounds for rectangular and cylindrical grid graphs.  More generally, we prove that transforming one immersion of $k$~circles with at most~$n$ self-crossings into another requires $\Theta(n^{3/2} + nk + k^2)$ homotopy moves in the worst case.  Finally, we prove that transforming one noncontractible closed curve to another on any orientable surface requires $\Omega(n^2)$ homotopy moves in the worst case; this lower bound is tight if the curve is homotopic to a simple closed curve.
\end{bigabstract}

\vfil
\begin{rightquote}{0.46}
“It's hardly fair,” muttered Hugh, “to give us such a jumble as this to work out!”
\quotee{Lewis Carroll, \emph{A Tangled Tale}, Knot X (1885)}
\end{rightquote}
\vfil
\vfil

\setcounter{page}{0}
\thispagestyle{empty}
\end{titlepage}

\newpage
\pagestyle{myheadings}
\markboth{Hsien-Chih Chang and Jeff Erickson}{Untangling Planar Curves}


\section{Introduction}

Any generic closed curve in the plane can be transformed into a simple closed curve by a finite sequence of the following local operations:
\begin{itemize}\itemsep0pt
\item \EMPH{$\arc{1}{0}$}: Remove an empty loop.
\item \EMPH{$\arc{2}{0}$}: Separate two subpaths that bound an empty bigon.
\item \EMPH{$\arc{3}{3}$}: Flip an empty triangle by moving one subpath over the opposite intersection point.
\end{itemize}
\begin{figure}[h]
\centering
\includegraphics[width=0.9\linewidth]{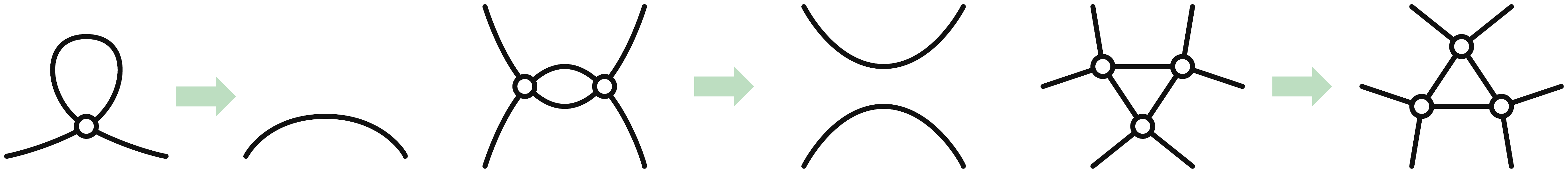}\\
\caption{Homotopy moves $\arc10$, $\arc20$, and $\arc33$.}
\label{F:homotopy}
\end{figure}
See Figure \ref{F:homotopy}. Each of these operations can be performed by continuously deforming the curve within a small neighborhood of one face; consequently, we call these operations and their inverses \EMPH{homotopy moves}. Our notation is nonstandard but mnemonic; the numbers before and after each arrow indicate the number of local vertices before and after the move. Homotopy moves are “shadows” of the classical Reidemeister moves used to manipulate knot and link diagrams~\cite{ab-tkc-26,r-ebk-27}.

We prove that $\Theta(n^{3/2})$ homotopy moves are sometimes necessary and always sufficient to simplify a closed curve in the plane with $n$ self-crossings.  Before describing our results in more detail, we review several previous results.

\subsection{Past Results}

An algorithm to simplify any planar closed curve using at most $O(n^2)$ homotopy moves is implicit in Steinitz's proof that every 3-connected planar graph is the 1-skeleton of a convex polyhedron \cite{s-pr-1916,sr-vtp-34}. Specifically, Steinitz proved that any non-simple closed curve (in fact, any 4-regular plane graph) with no empty loops contains a \emph{bigon} (“Spindel”): a disk bounded by a pair of simple subpaths that cross exactly twice, where the endpoints of the (slightly extended) subpaths lie outside the disk.  Steinitz then proved that any \emph{minimal} bigon (“irreduzible Spindel”) can be transformed into an empty bigon using a sequence of $\arc33$ moves, each removing one triangular face from the bigon, as shown in Figure~\ref{F:SR-spindle}.  Once the bigon is empty, it can be deleted with a single $\arc20$ move. See Grünbaum~\cite{g-cp-67}, Hass and Scott \cite{hs-scs-94}, Colin de Verdière \etal~\cite{cgv-rep-96}, or Nowik~\cite{n-cpsc-09} for more modern treatments of Steinitz's technique. The $O(n^2)$ upper bound also follows from algorithms for \emph{regular} homotopy, which forbids $\biarc01$ moves, by Francis~\cite{f-frtcs-69}, Vegter~\cite{v-kfdp-89} (for polygonal curves), and Nowik~\cite{n-cpsc-09}.

\begin{figure}[ht]
\centering
\includegraphics[scale=0.175]{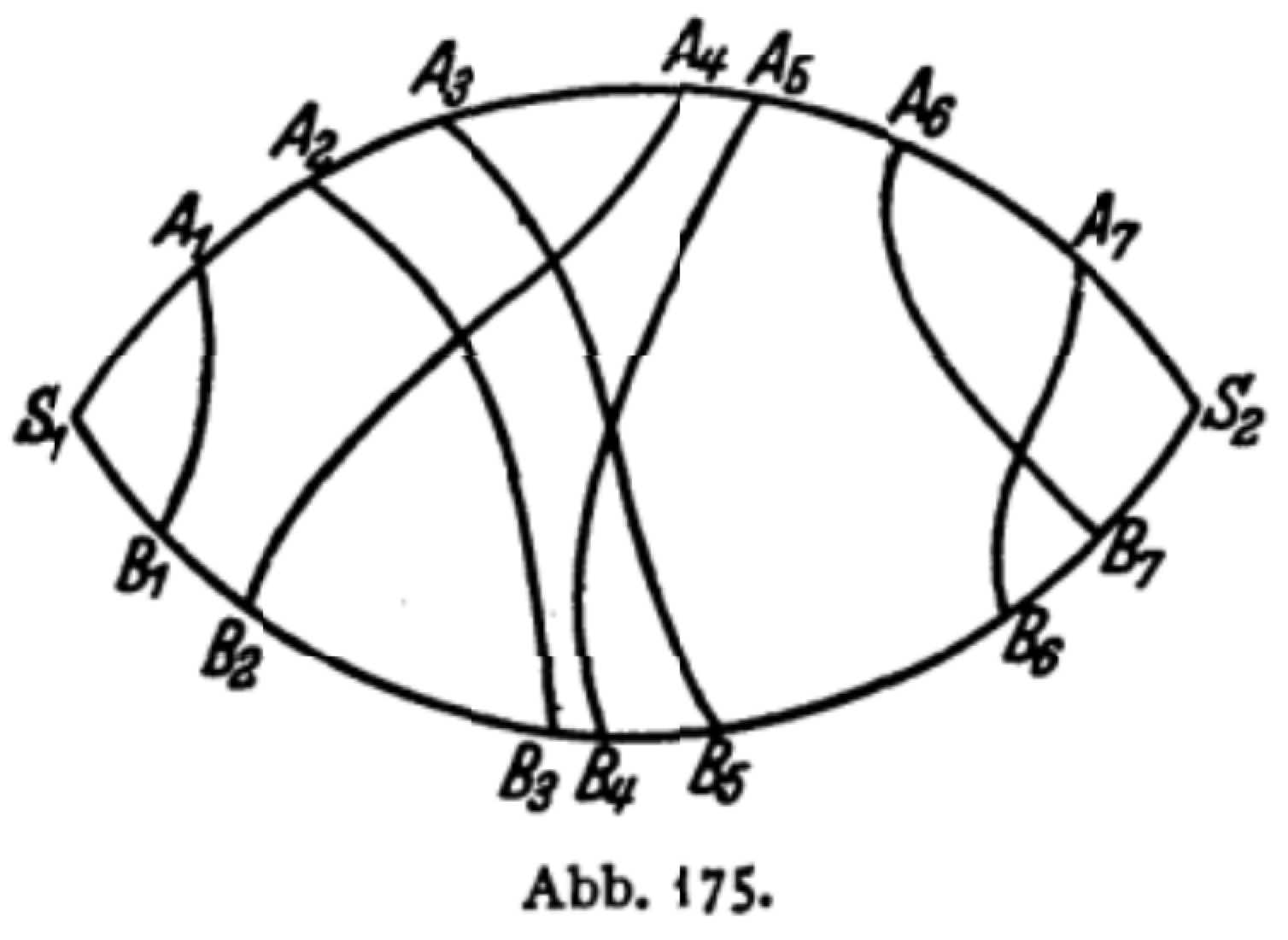} \\
\includegraphics[scale=0.1]{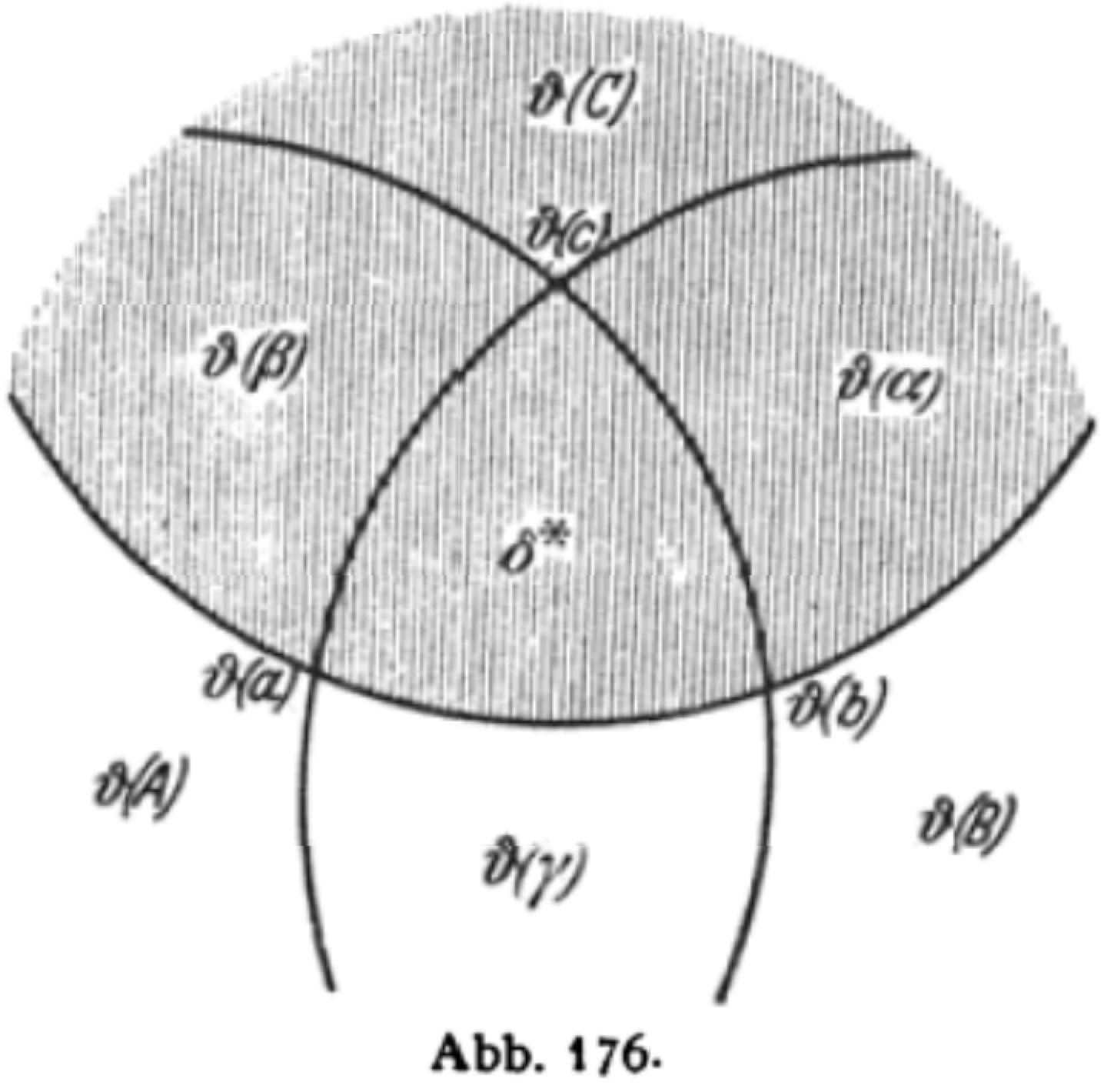} \quad
\includegraphics[scale=0.1]{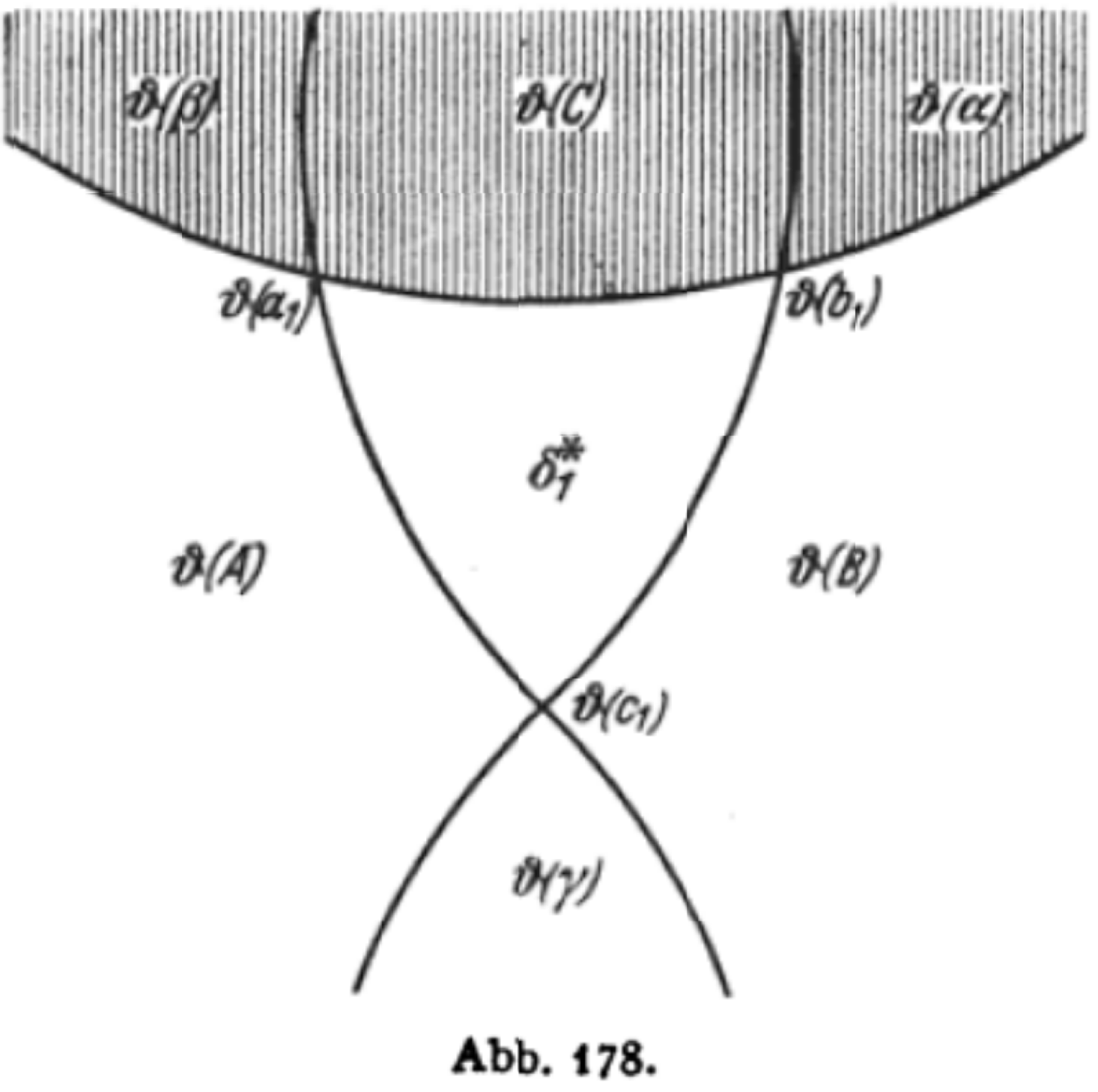} \qquad
\includegraphics[scale=0.12]{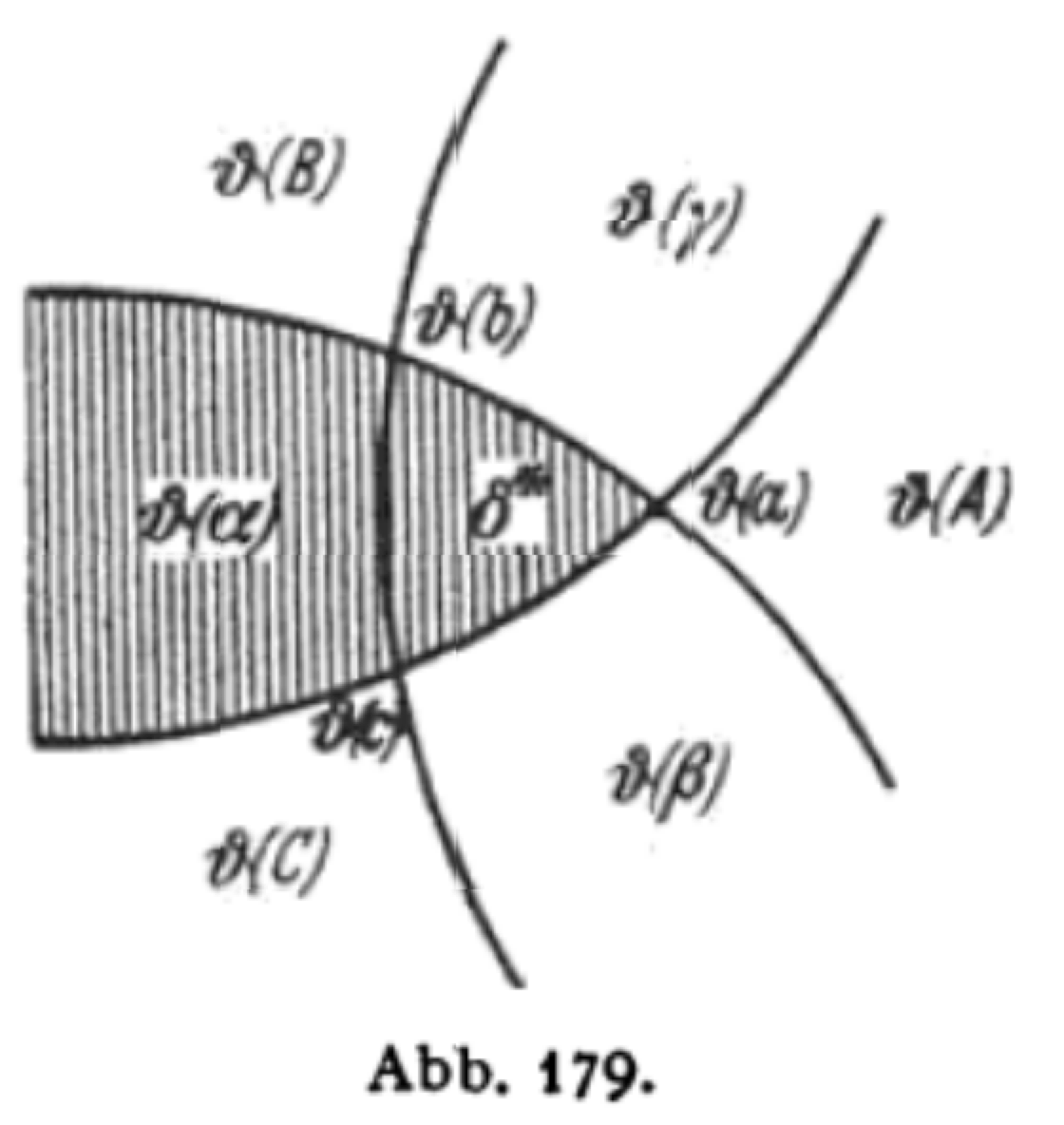} \quad
\includegraphics[scale=0.12]{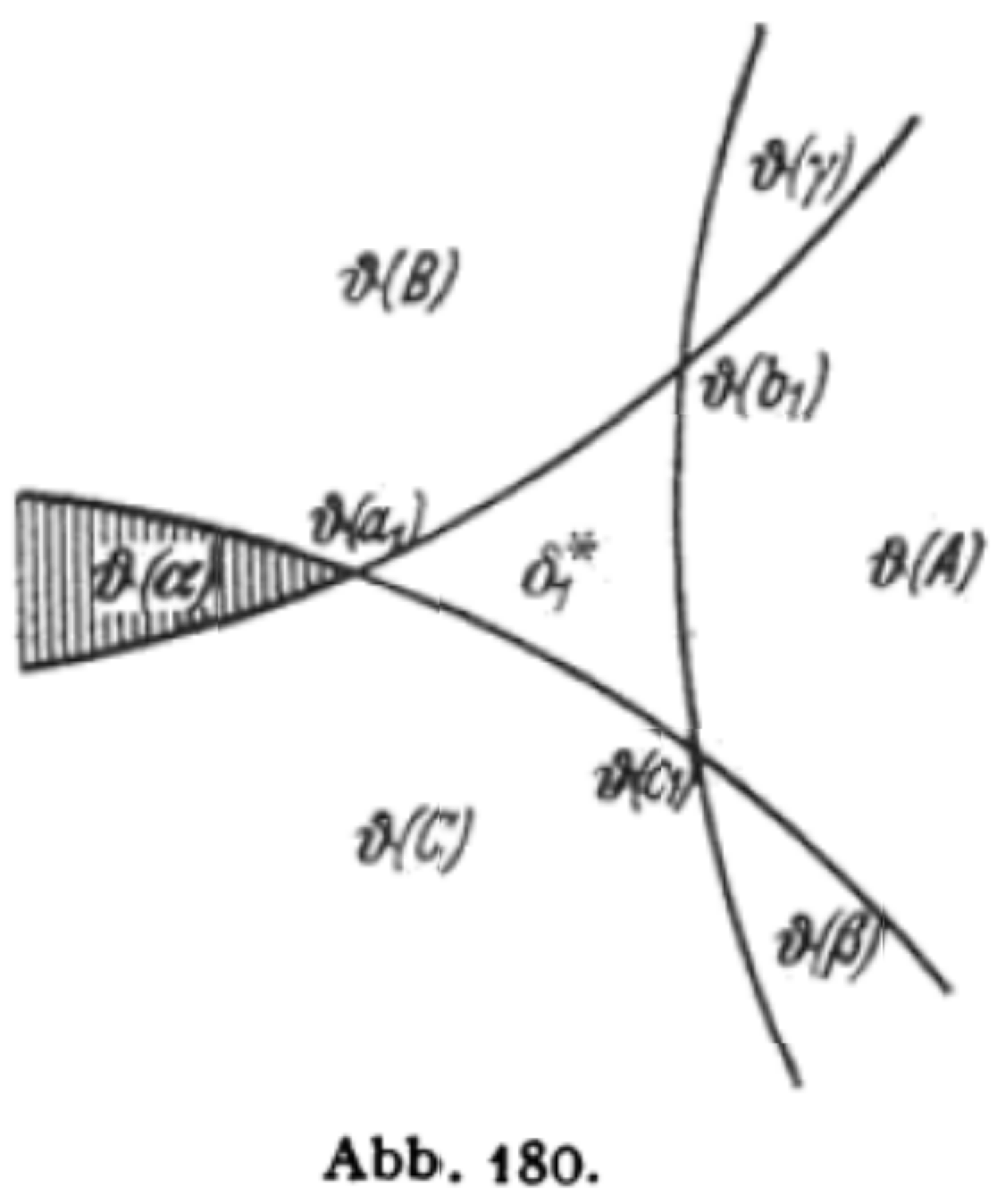}
\caption{Top: A minimal bigon.  Bottom: $\arc33$ moves removing triangles from the side or the end of a (shaded) minimal bigon.  All figures are from Steinitz and Rademacher \cite{sr-vtp-34}.}
\label{F:SR-spindle}
\end{figure}

The $O(n^2)$ upper bound can also be derived from an algorithm of Feo and Provan~\cite{fp-dtert-93} for reducing a \emph{plane} graph to a single edge by \emph{electrical transformations}: degree-1 reductions, series-parallel reductions, and $\Delta$Y-transformations. (We consider electrical transformations in more detail in Section \ref{S:electric}.) Any curve divides the plane into regions, called its \emph{faces}. The \emph{depth} of a face is its distance to the outer face in the dual graph of the curve. Call a homotopy move \emph{positive} if it decreases the sum of the face depths; in particular, every $\arc10$ and $\arc20$ move is positive.  \EDIT{A key technical lemma} of Feo and Provan implies that every non-simple curve in the plane admits a positive homotopy move \cite[Theorem~1]{fp-dtert-93}. Thus, the sum of the face depths is an upper bound on the minimum number of moves required to simplify the curve. Euler's formula implies that every curve with $n$ crossings has $O(n)$ faces, and each of these faces has depth $O(n)$.

Gitler~\cite{g-dtaa-91} conjectured that a variant of Feo and Provan's algorithm that always makes the \emph{deepest} positive move requires only $O(n^{3/2})$ moves. Song~\cite{s-iifpd-01} observed that if Feo and Provan's algorithm always chooses the \emph{shallowest} positive move, it can be forced to make $\Omega(n^2)$ moves even when the input curve can be simplified using only $O(n)$ moves.

Tight bounds are known for two special cases where some homotopy moves are forbidden. First, Nowik~\cite{n-cpsc-09} proved a tight $\Omega(n^2)$ lower bound for regular homotopy. Second, Khovanov~\cite{k-dg-97} defined two curves to be \emph{doodle equivalent} if one can be transformed into the other using $\biarc10$ and $\biarc20$ moves. Khovanov~\cite{k-dg-97} and Ito and Takimura~\cite{it-whkp-13} independently proved that any planar curve can be transformed into its unique equivalent doodle with the smallest number of vertices, using only $\arc10$ and $\arc20$ moves. Thus, two doodle equivalent curves are connected by a sequence of $O(n)$ moves, which is obviously tight.

Looser bounds are also known for the minimum number of Reidemeister moves needed to reduce a diagram of the unknot~\cite{hn-udrqn-10,l-pubrm-15}, to separate the components of a split link~\cite{hhn-unnrm-12}, or to move between two equivalent knot diagrams~\cite{hh-msrmd-11,cl-ubrm-14}.

\subsection{New Results}

In Section \ref{S:lower-bound}, we derive an $\Omega(n^{3/2})$ lower bound using a numerical curve invariant called \emph{defect}, introduced by Arnold~\cite{a-tipcc-94, a-pctip-94} and~Aicardi~\cite{a-tc-94}.  Each homotopy move changes the defect of a closed curve by at most~$2$.  The lower bound therefore follows from constructions of Hayashi \etal~\cite{hh-msrmd-11,hhsy-musrm-12} and Even-Zohar \etal~\cite{ehln-irkl-14} of closed curves with defect $\Omega(n^{3/2})$.  We simplify and generalize their results by computing the defect of the standard planar projection of any $p\times q$ torus knot where either $p\bmod q = 1$ or $q\bmod p = 1$. Our calculations imply that for any integer $p$, reducing the standard projection of the $p\times (p+1)$ torus knot requires at least $\smash{\binom{p+1}{3}} \ge n^{3/2}/6 - O(n)$ homotopy moves.  Finally, using winding-number arguments, we prove that in the worst case, simplifying an arrangement of $k$ closed curves requires $\Omega(n^{3/2}+ nk)$ homotopy moves, with an additional $\Omega(k^2)$ term if the target configuration is specified in advance.


In Section \ref{S:electric}, we provide a proof, based on arguments of Truemper~\cite{t-drpg-89} and Noble and Welsh~\cite{nw-kg-00}, that reducing a \emph{unicursal} \EDIT{plane} graph $G$---\EDIT{one} whose medial graph is the image of a single closed curve---using \EDIT{facial} electrical transformations requires at least as many steps as reducing the medial graph of~$G$ to a simple closed curve using homotopy moves.  The homotopy lower bound from Section \ref{S:lower-bound} then implies that reducing any $n$-vertex \EDIT{plane} graph with treewidth $\Omega(\sqrt{n})$ requires $\Omega(n^{3/2})$ \EDIT{facial} electrical transformations.  This lower bound matches known upper bounds for rectangular and cylindrical grid graphs.  

We develop a new algorithm to simplify any closed curve in $O(n^{3/2})$ homotopy moves in Section~\ref{S:upper}. First we describe an algorithm that uses $O(D)$ moves, where $D$ is the sum of the face depths of the input curve. At a high level, our algorithm can be viewed as a variant of Steinitz's algorithm that empties and removes \emph{loops} instead of bigons.
%
%
We then extend our algorithm to \emph{tangles}: collections of boundary-to-boundary paths in a closed disk. Our algorithm simplifies a tangle as much as possible in $O(D + ns)$ moves, where $D$ is the sum of the depths of the tangle's faces, $s$ is the number of paths, and~$n$ is the number of intersection points.
Then, we prove that for any curve with maximum face depth $\Omega(\sqrt{n})$, we can find a simple closed curve whose interior tangle has \EDIT{$m$ interior vertices, at most $\sqrt{m}$ paths, and maximum face depth $O(\sqrt{n})$}.
Simplifying this tangle and then recursively simplifying the resulting curve requires a total of $O(n^{3/2})$ moves.
We show that this simplifying sequence of homotopy moves can be computed in $O(1)$ amortized time per move, assuming the curve is presented in an appropriate graph data structure.
We conclude this section by proving that any arrangement of $k$ closed curves can be simplified in $O(n^{3/2} + nk)$ homotopy moves, or in $O(n^{3/2} + nk + k^2)$ homotopy moves if the target configuration is specified in advance, precisely matching our lower bounds for all values of $n$ and $k$.

Finally, in Section \ref{S:genus}, we consider curves on surfaces of higher genus. We prove that $\Omega(n^2)$ homotopy moves are required in the worst case to transform one non-contractible closed curve to another on the torus, and therefore on any orientable surface. Results of Hass and Scott \cite{hs-ics-85} imply that this lower bound is tight if the non-contractible closed curve is homotopic to a simple closed curve.

\subsection{Definitions}

A \EMPH{closed curve} in a surface $M$ is a continuous map $\gamma \colon S^1 \to M$.  In this paper, we consider only \emph{generic} closed curves, which are injective except at a finite number of self-intersections, each of which is a transverse double point; closed curves satisfying these conditions are called \emph{immersions} of the circle.  A~closed curve is \EMPH{simple} if it is injective. For most of the paper, we consider only closed curves in the plane; we consider more general surfaces in Section \ref{S:genus}.

The image of any non-simple closed curve has a natural structure as a 4-regular plane graph.  Thus, we refer to the self-intersection points of a curve as its \EMPH{vertices}, the maximal subpaths between vertices as \EMPH{edges}, and the components of the complement of the curve as its \EMPH{faces}.  Two curves $\gamma$ and $\gamma'$ are \emph{isomorphic} if their images \EDIT{define combinatorially equivalent maps}; we will not distinguish between isomorphic curves.

A \emph{corner} of $\gamma$ is the intersection of a face of $\gamma$ and a small neighborhood of a vertex of $\gamma$.  A \EMPH{loop} in a closed curve $\gamma$ is a subpath of $\gamma$ that begins and ends at some vertex $x$, intersects itself only at $x$, and encloses exactly one corner at $x$.  A \EMPH{bigon} in $\gamma$ consists of two simple interior-disjoint subpaths of $\gamma$ with the same endpoints and enclose one corner at each of those endpoints.  A loop or bigon is \EMPH{empty} if its interior does not intersect $\gamma$.  Notice that a $\arc10$ move is applied to an empty loop, and a $\arc20$ move is applied on an empty bigon.

We adopt a standard sign convention for vertices first used by Gauss~\cite{g-n1gs-00}. Choose an arbitrary basepoint $\gamma(0)$ and orientation for the curve.  \EDIT{For each vertex $x$,} we define $\sgn(x) = +1$ if the first traversal through the vertex crosses the second traversal from right to left, and $\sgn(x) = -1$ otherwise.  See Figure~\ref{F:signs}.

\begin{figure}[ht]
\centering
\includegraphics[scale=0.4]{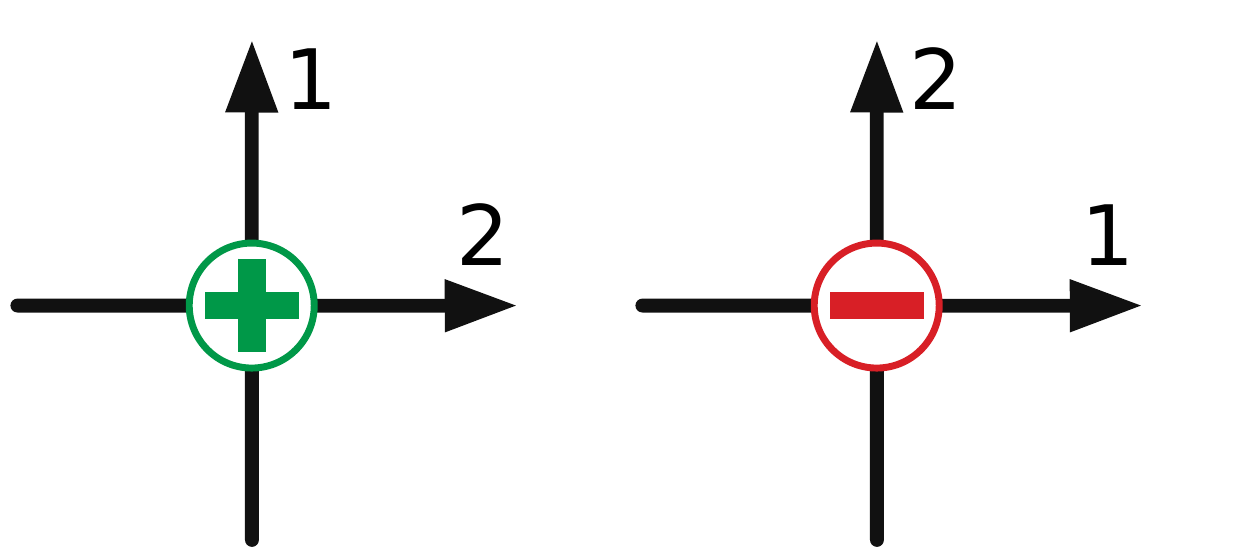}
\caption{Gauss's sign convention.}
\label{F:signs}
\end{figure}

A \EMPH{homotopy} between two curves $\gamma$ and $\gamma'$ in surface $M$ is a continuous function $H\colon {S^1 \times [0,1] \to M}$ such that $H(\cdot,0) = \gamma$ and $H(\cdot,1) = \gamma'$.  Any homotopy $H$ describes a continuous deformation \EDIT{from} $\gamma$ \EDIT{to}~$\gamma'$, where the second argument of $H$ is “time”.  Each homotopy move can be executed by a homotopy. Conversely, Alexander's simplicial approximation theorem \cite{a-cas-26}, together with combinatorial arguments of Alexander and Briggs~\cite{ab-tkc-26} and Reidemeister~\cite{r-ebk-27}, imply that any generic homotopy between two closed curves can be decomposed into a finite sequence of homotopy moves. Two curves are \emph{homotopic}, or in the same \emph{homotopy class}, if there is a homotopy from one to the other. All closed curves in the plane are homotopic.

A \EMPH{multicurve} is an immersion of one or more disjoint circles; in particular, a \EMPH{$k$-curve} is an immersion of $k$ disjoint circles.  A multicurve is \emph{simple} if it is injective, or equivalently, if it can be decomposed into pairwise disjoint simple closed curves.  The image of any multicurve in the plane is the disjoint union of simple closed curves and 4-regular plane graphs.   A \EMPH{component} of a multicurve $\gamma$ is any multicurve whose image is a connected component of the image of $\gamma$.  We call the individual closed curves that comprise a multicurve its \EMPH{constituent \EDIT{curves}}; see Figure \ref{F:components}.  The definition of homotopy and the decomposition of homotopies into homotopy moves extend naturally to multicurves.  

\begin{figure}[ht]
\centering
\includegraphics[scale=0.4]{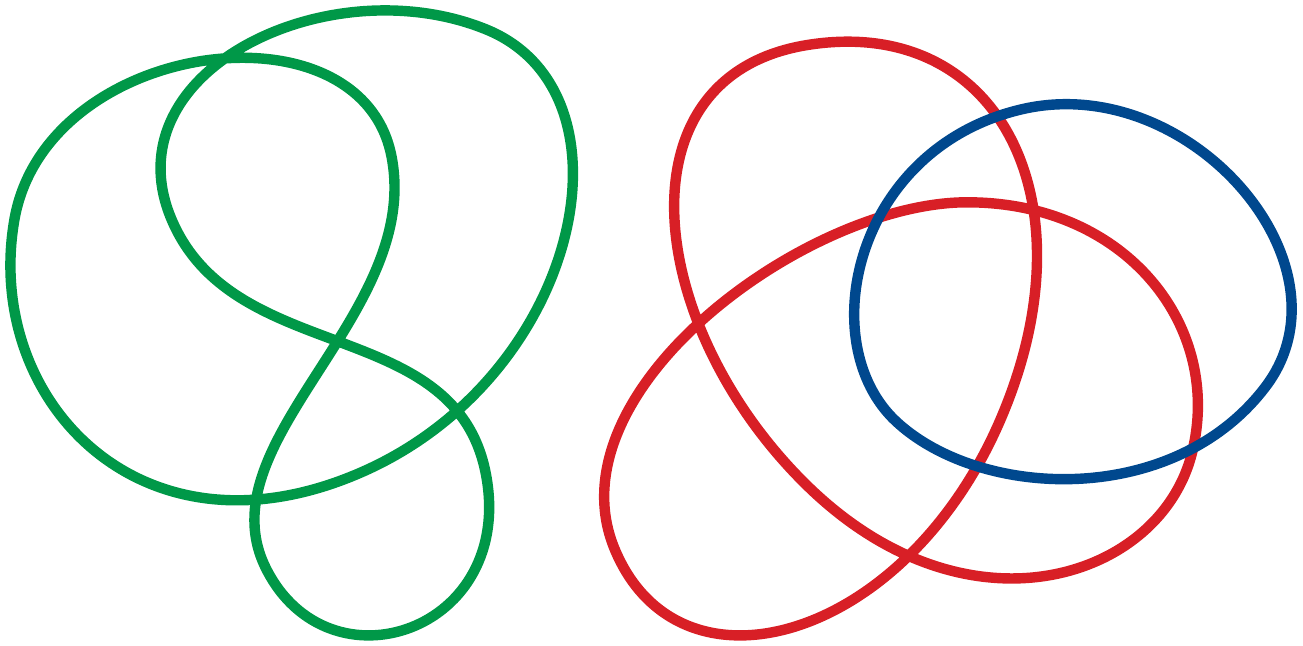}
\caption{A multicurve with two components and three constituent curves, one of which is simple.}
\label{F:components}
\end{figure}


\section{Lower Bounds}
\label{S:lower-bound}

\subsection{Defect}
\label{SS:defect}

To prove our main lower bound, we consider a numerical invariant of closed curves in the plane introduced by Arnold~\cite{a-tipcc-94, a-pctip-94} and~Aicardi~\cite{a-tc-94} called \EMPH{defect}. Polyak~\cite{p-icfgd-98} proved that defect can be computed---or for our purposes, defined---as follows:
\[
	\Defect(\gamma) \coloneqq -2 \sum_{x\between y} \sgn(x)\cdot\sgn(y).
\]
Here the sum is taken over all \emph{interleaved} pairs of vertices of $\gamma$: two vertices $x\ne y$ are interleaved, denoted \EMPH{$x\between y$}, if they alternate in cyclic order---$x$, $y$, $x$, $y$---along $\gamma$.
(The factor of $-2$ is a historical artifact, which we retain only to be consistent with Arnold's original definitions~\cite{a-tipcc-94, a-pctip-94}.)
Even though the signs of individual vertices depend on the basepoint and orientation of the curve, the defect of a curve is independent of those choices. Moreover, the defect of any curve is preserved by any homeomorphism from the plane (or the sphere) to itself, including reflection. 


%
%

Trivially, every simple closed curve has defect zero. Straightforward case analysis~\cite{p-icfgd-98} implies that any single homotopy move changes the defect of a curve by at most $2$; the various cases are listed below and illustrated in Figure \ref{F:defect-change}.  
\vspace{3pt}
\begin{itemize}\itemsep0pt
\item A $\arc10$ move leaves the defect unchanged.
\item A $\arc20$ move decreases the defect by $2$ if the two disappearing vertices are interleaved, and leaves the defect unchanged otherwise.
\item A $\arc33$ move increases the defect by $2$ if the three vertices before the move contain an even number of interleaved pairs, and decreases the defect by $2$ otherwise.
\end{itemize}
In light of this case analysis, the following lemma is trivial:
\begin{lemma}
\label{L:defect}
Simplifying any closed curve $\gamma$ in the plane requires at least $\abs{\Defect(\gamma)}/2$ homotopy moves.
\end{lemma}
\unskip

\begin{figure}[htb]
\centering\small
\def\arraystretch{1.25}
\begin{tabular}{cc@{~}cc@{~}c}
	\toprule
	$1\arcto 0$ & \multicolumn{2}{c}{$2\arcto 0$} & \multicolumn{2}{c}{$3\arcto 3$}
	\\ \midrule
	\raisebox{-.5\height}{\includegraphics[scale=0.25]{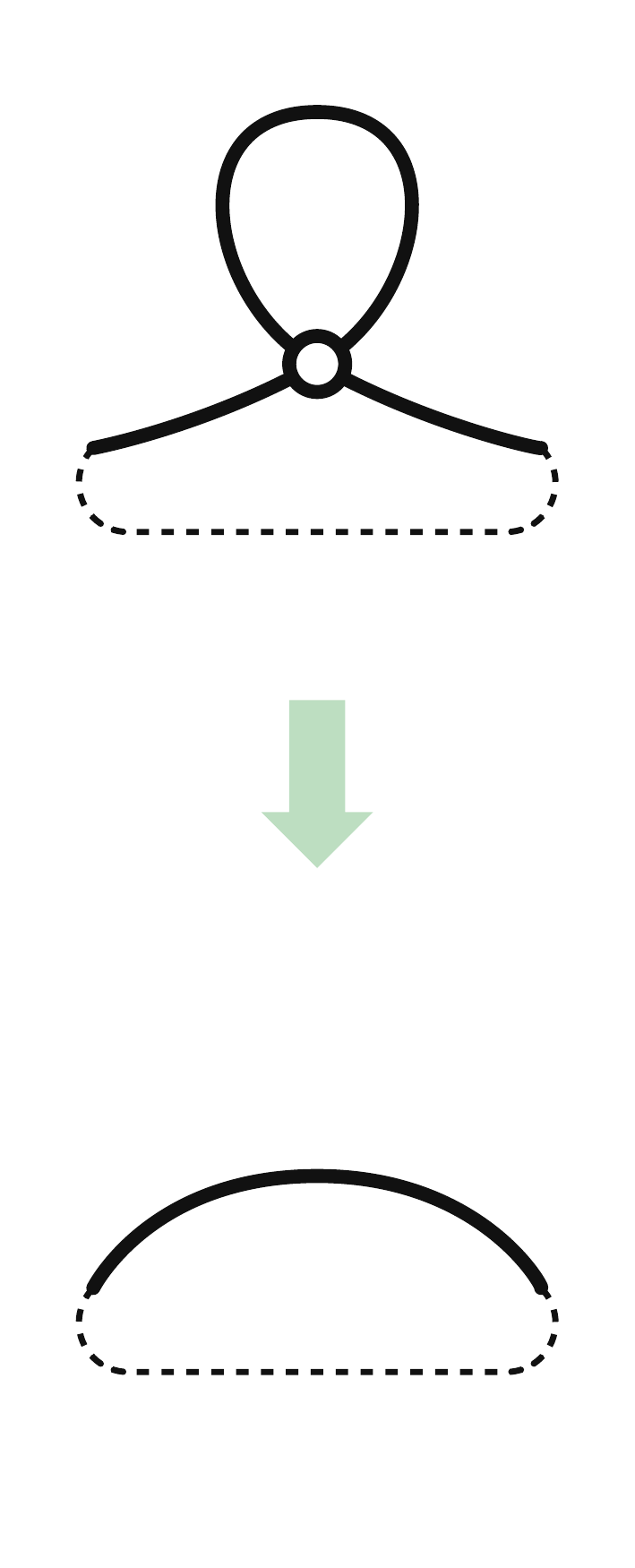}}
	&
	\raisebox{-.5\height}{\includegraphics[scale=0.25]{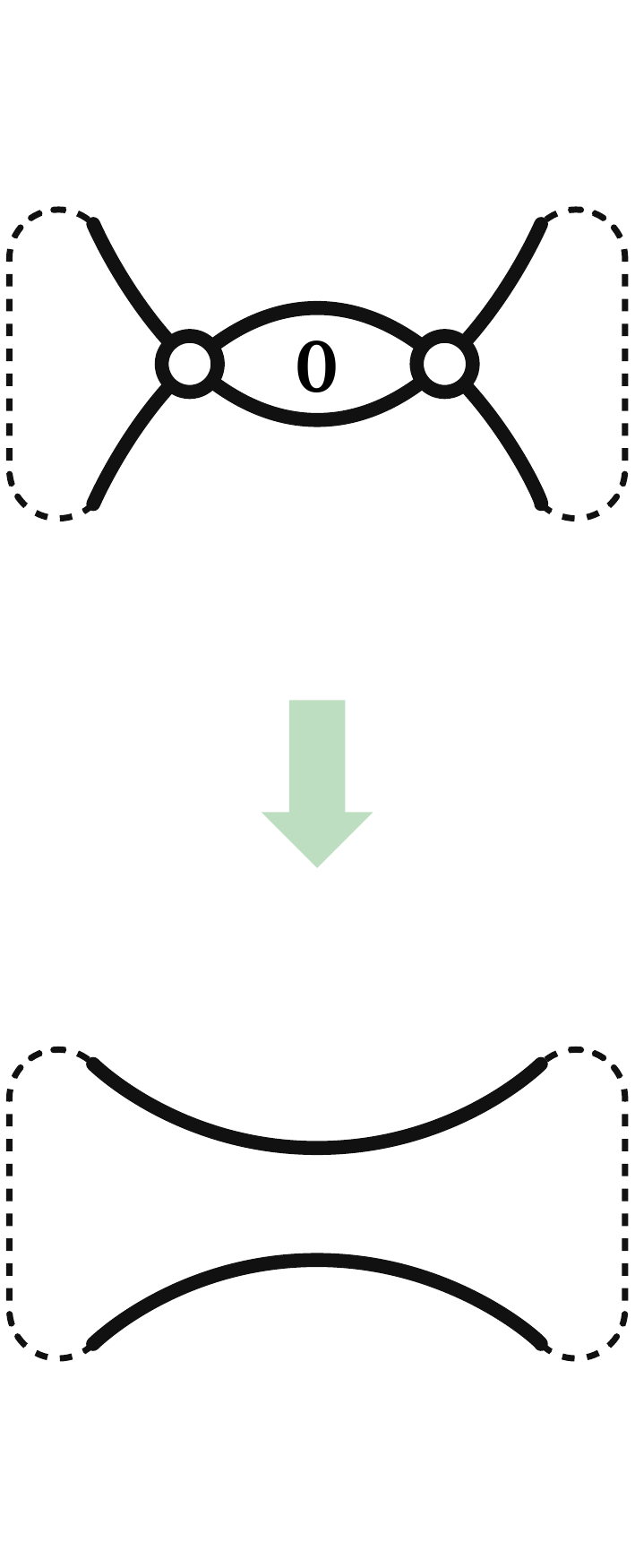}}
	&
	\raisebox{-.5\height}{\includegraphics[scale=0.25]{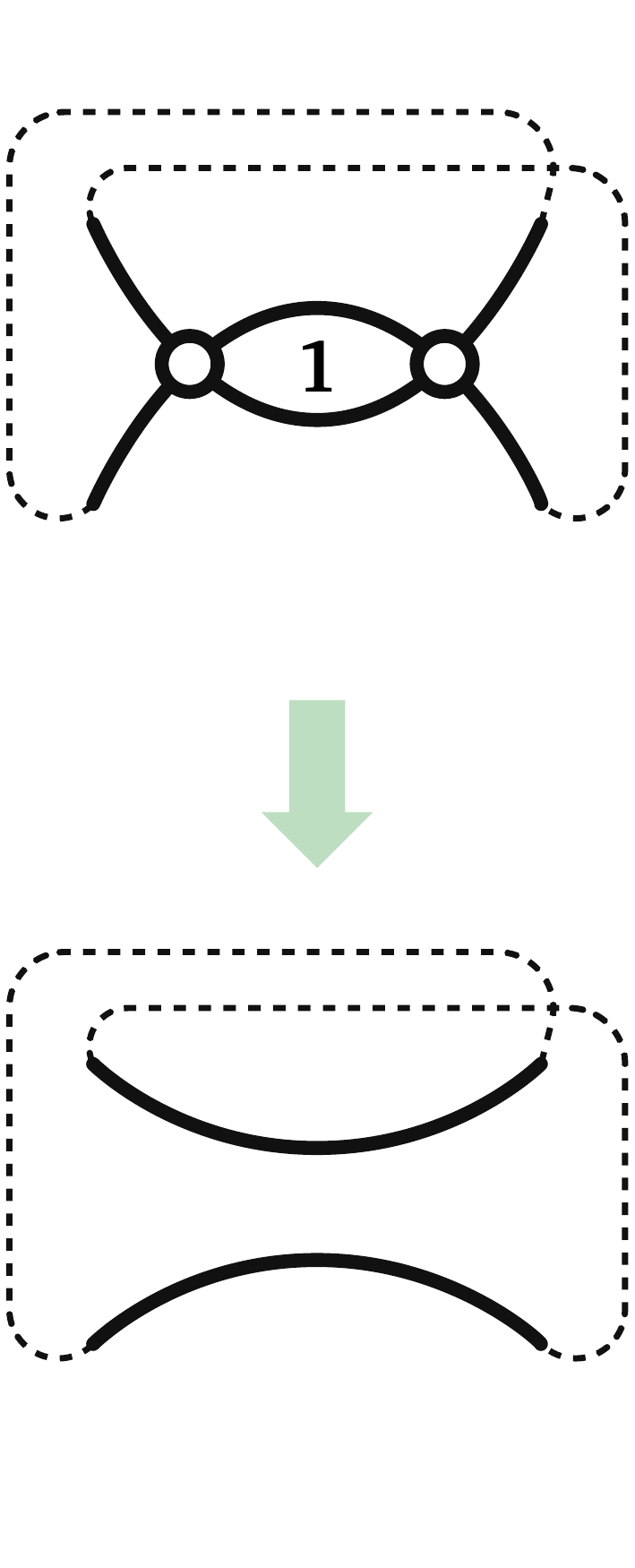}}
	&
	\raisebox{-.5\height}{\includegraphics[scale=0.25]{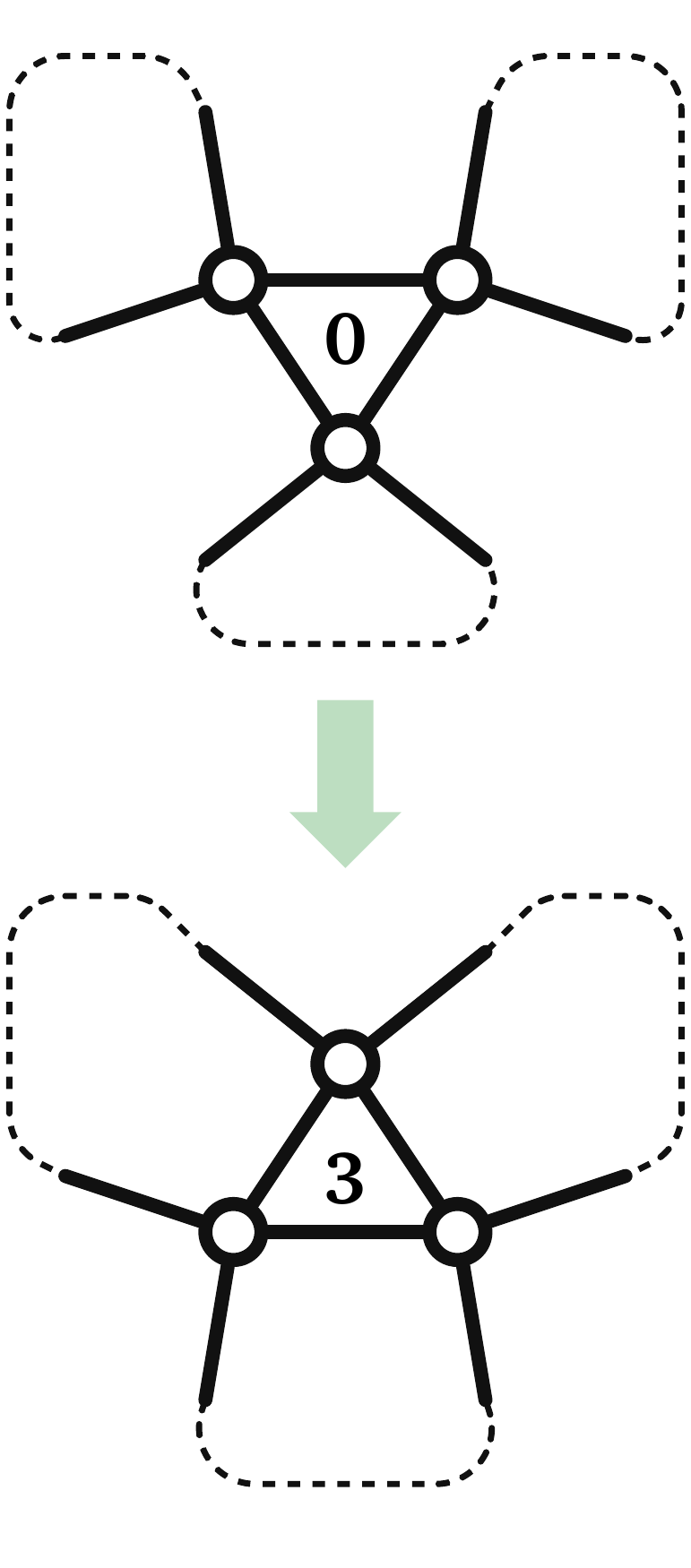}}
	&
	\raisebox{-.5\height}{\includegraphics[scale=0.25]{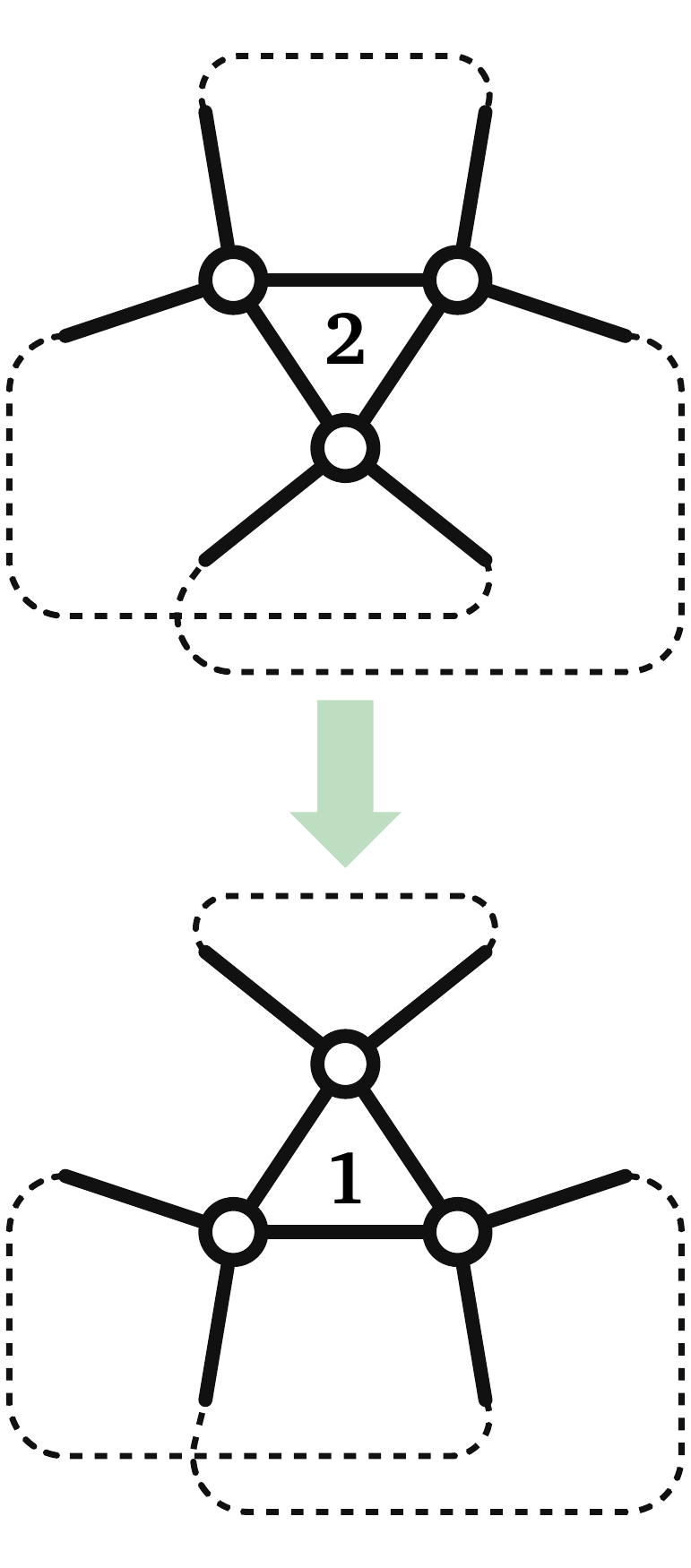}}
\\ \midrule
	$0$ & $0$ & $-2$ & $+2$ & $+2$
	\\ \bottomrule
\end{tabular}
\caption{Changes to defect incurred by homotopy moves. Numbers in each figure indicate how many pairs of vertices are interleaved; dashed lines indicate how the rest of the curve connects.}
\label{F:defect-change}
\end{figure}

\subsection{Flat Torus Knots}
\label{SS:torus-knots}

For any relatively prime positive integers $p$ and $q$, let \EMPH{$T(p,q)$} denote the curve with the following parametrization, where $\theta$ runs from~$0$ to $2\pi$:
\[
	T(p,q)(\theta) \coloneqq \left((\cos (q\theta)+2) \cos(p\theta),~ (\cos (q\theta)+2) \sin(p\theta)\right).
\]
The curve $T(p,q)$ winds around the origin $p$ times, oscillates $q$ times between two concentric circles, and crosses itself exactly $(p-1)q$ times.  We call these curves \EMPH{flat torus knots}.

\begin{figure}[ht]
\centering
\hfil
\includegraphics[width=1.75in]{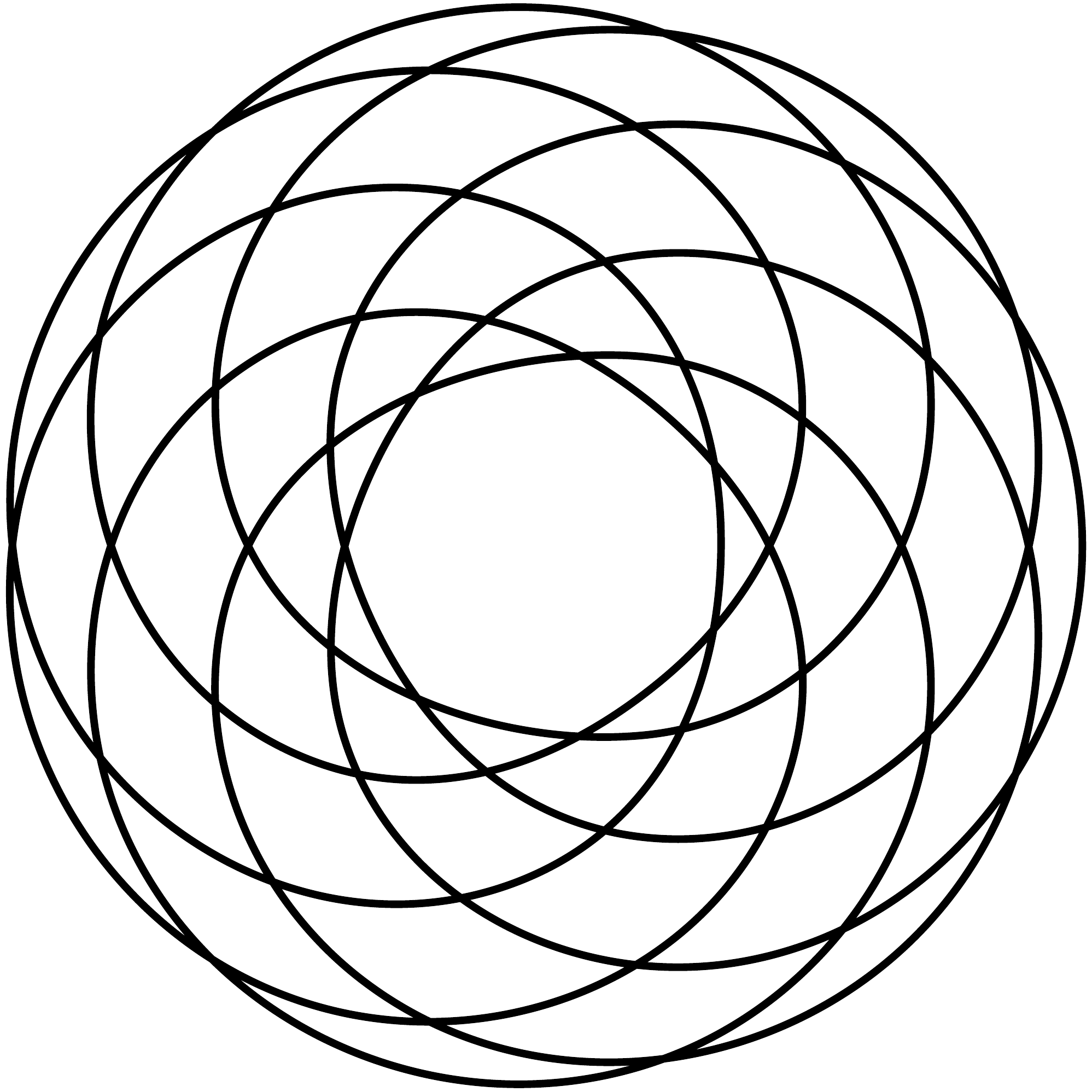}\hfil{}
\includegraphics[width=1.75in]{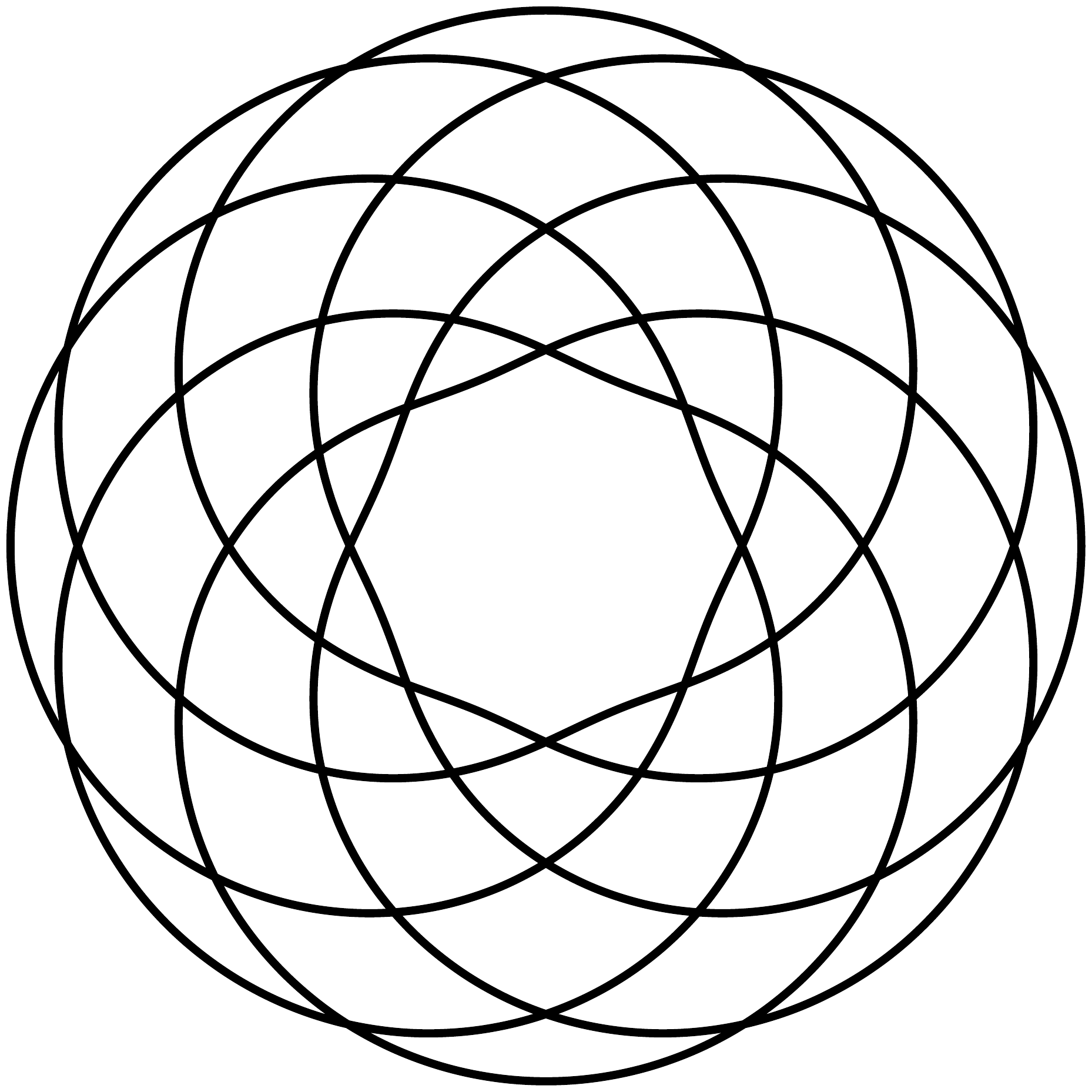}\hfil{}
\caption{The flat torus knots $T(8,7)$ and $T(7,8)$.}
\end{figure}

Hayashi \etal~\cite[Proposition~3.1]{hhsy-musrm-12} proved that for any integer~$q$, the flat torus knot $T(q+1,q)$ has defect $-2\binom{q}{3}$. Even-Zohar \etal~\cite{ehln-irkl-14} used a star-polygon representation of the curve $T(p, 2p+1)$ as the basis for a universal model of random knots; in our notation, they proved that $\Defect(T(p, 2p+1)) = 4\binom{p+1}{3}$ for any integer $p$. 
In this section we simplify and generalize both of these results to all flat torus knots $T(p,q)$ where either $q\bmod p = 1$ or $p\bmod q = 1$.  \EDIT{For purposes of illustration, we cut $T(p,q)$ along a spiral path parallel to a portion of the curve, and then deform the $p$ resulting subpaths, which we call \emph{strands}, into a “flat braid” between two fixed diagonal lines.  See Figure~\ref{F:flat-braid}.}

\begin{figure}[ht]
\centering
\includegraphics[scale=0.3]{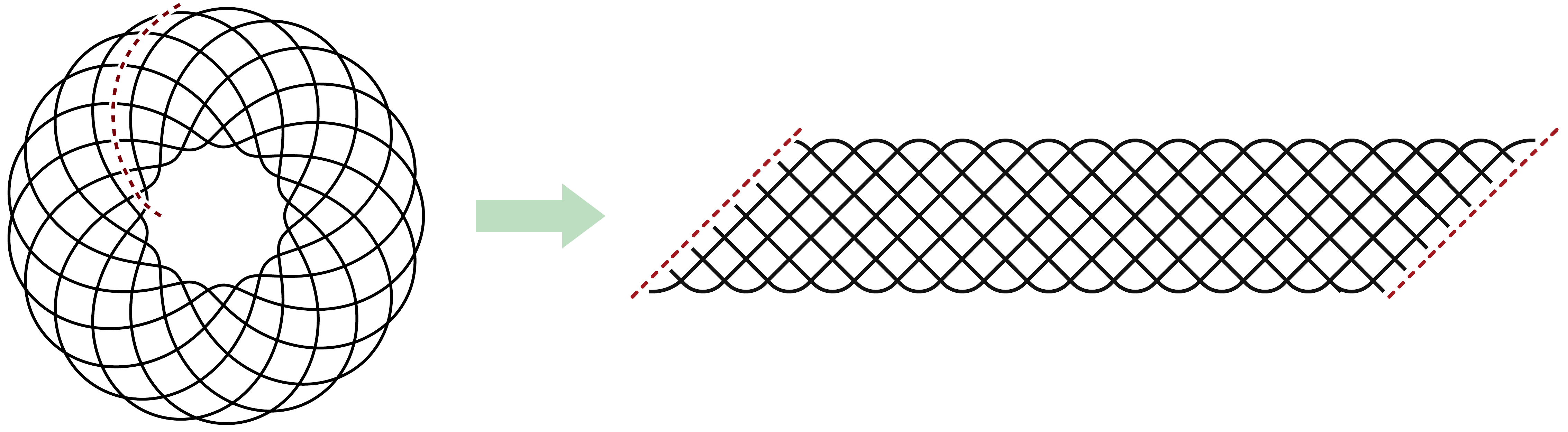}
\caption{Transforming $T(8,17)$ into a flat braid.}
\label{F:flat-braid}
\end{figure}

\begin{lemma}
\label{L:braid-wide}
$\Defect(T(p, ap+1)) = 2a \binom{p+1}{3}$ for all integers $a\ge 0$ and $p \ge 1$.
\end{lemma}

\begin{proof}
The curve $T(p, 1)$ can be reduced \EDIT{to a simple closed curve} using only $\arc10$ moves, so its defect is zero.  \EDIT{For the rest of the proof, assume $a\ge 1$.}

\EDIT{We define a \emph{stripe} of $T(p,ap+1)$ to be a subpath from some innermost point to the next outermost point, or equivalently, a subpath of any strand from the bottom to the top in the flat braid representation.  Each stripe contains exactly $p-1$ crossings.  A~\emph{block} of $T(p,ap+1)$ consists of $p(p-1)$ crossings in $p$ consecutive stripes; within any block, each pair of strands intersects exactly twice.}  We can reduce $T(p, ap+1)$ to $T(p, ({a-1})p+1)$ by straightening \EDIT{any block} one strand at a time.  Straightening the bottom strand of \EDIT{the} block requires the following $\binom{p}{2}$ moves, as shown in Figure \ref{F:braid-wide}.

\begin{itemize}
\item
$\binom{p-1}{2}$ $\arc33$ moves pull the bottom strand downward over one intersection point of every other pair of strands. Just before each $\arc33$ move, exactly one of the three pairs of the three relevant vertices is interleaved, so each move decreases the defect by $2$.

\item
$(p-1)$ $\arc20$ moves eliminate a pair of intersection points between the bottom strand and every other strand. Each of these moves also decreases the defect by $2$.
\end{itemize}
Altogether, straightening one strand decreases the defect by $\smash{2\binom{p}{2}}$. Proceeding similarly with the other strands, 
we conclude that $\Defect(T(p, ap+1)) = \Defect(T(p, {(a-1)p+1})) + 2\binom{p+1}{3}$. The lemma follows immediately by induction.
\end{proof}

\begin{figure}[ht]
\centering
\includegraphics[scale=0.275]{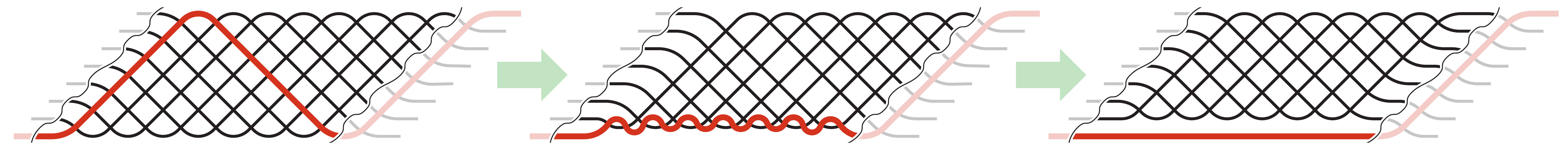}
\caption{Straightening one strand in a block of $T(8, 8a+1)$.}
\label{F:braid-wide}
\end{figure}

\vspace{3pt}

\begin{lemma}
\label{L:braid-deep}
$\Defect(T(aq+1, q)) = -2a \binom{q}{3}$ for all integers $a\ge 0$ and $q \ge 1$.
\end{lemma}

\begin{proof}
The curve $T(1,q)$ is simple, so its defect is trivially zero. For any positive integer $a$, we can transform $T(aq+1, q)$ into $T((a-1)q+1, q)$ by incrementally removing the innermost $q$ \emph{loops}. We can remove the first loop using $\binom{q}{2}$ homotopy moves, as shown in Figure \ref{F:braid-deep}. (The first transition in Figure~\ref{F:braid-deep} just reconnects the top left and top right endpoints of the flat braid.)
\begin{itemize}
\item
$\binom{q-1}{2}$ $\arc33$ moves pull the left side of the loop to the right, over the crossings inside the loop. Just before each $\arc33$ move, the three relevant vertices contain two interleaved pairs, so each move \emph{increases} the defect by $2$.
\item
$(q-1)$ $\arc20$ moves pull the loop over $q-1$ strands. The strands involved in each move are oriented in opposite directions, so these moves leave the defect unchanged.
\item
Finally, we can remove the loop with a single $\arc10$ move, which does not change the defect.
\end{itemize}
Altogether, removing one loop increases the defect by $\smash{2\binom{q-1}{2}}$. Proceeding similarly with the other loops, 
we conclude that $\Defect(T(aq+1, q)) = \Defect(T((a-1)q+1, q)) - 2 \binom{q}{3}$. The lemma follows immediately by induction.
\end{proof}

\begin{figure}[ht]
\centering
\includegraphics[scale=0.275]{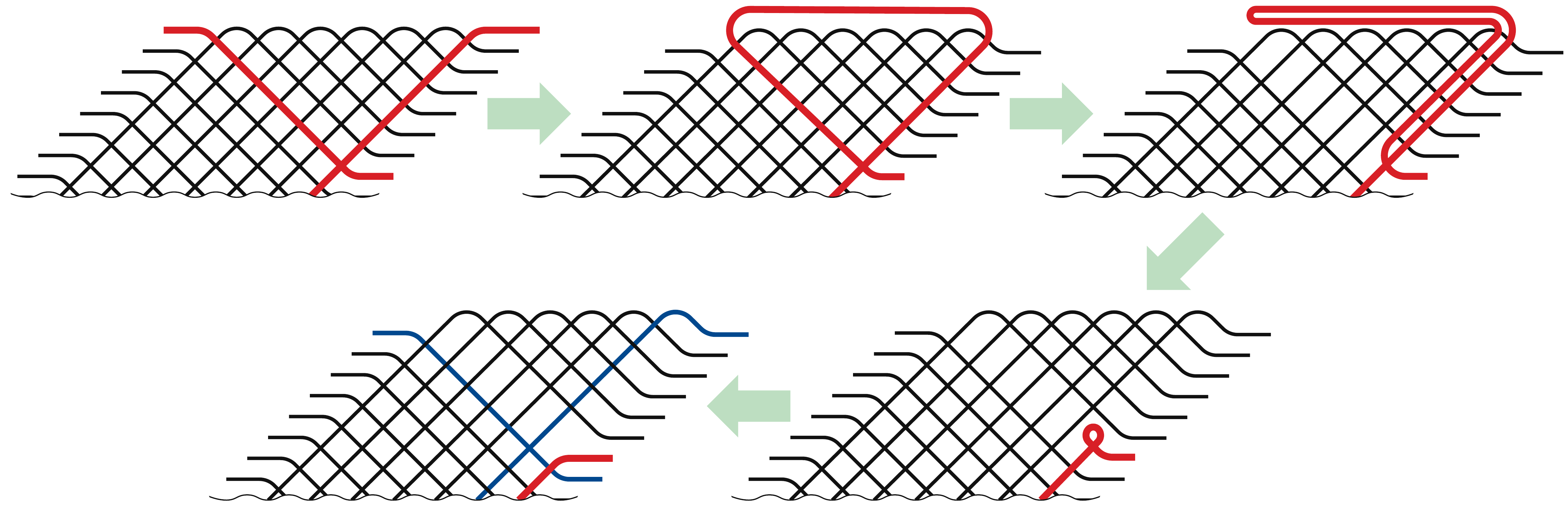}
\caption{Removing one loop from the innermost block of $T(7a+1, 7)$.}
\label{F:braid-deep}
\end{figure}

Either of the previous lemmas imply the following lower bound, which is also implicit in the work of Hayashi \etal~\cite{hhsy-musrm-12}.

\begin{theorem}
For every positive integer~$n$, there are closed curves with $n$ vertices whose defects are $n^{3/2}/3 - O(n)$ and $-n^{3/2}/3 + O(n)$, and therefore requires at least $n^{3/2}/6 - O(n)$ homotopy moves to reduce to a simple closed curve.
\end{theorem}

\begin{proof}
The lower bound follows from the previous lemmas by setting $a=1$.  If $n$ is a prefect square, then the flat torus knot $T(\sqrt{n}+1, \sqrt{n})$ has $n$ vertices and defect $\smash{-2\binom{\sqrt{n}}{3}}$.  If $n$ is not a perfect square, we can achieve defect ${-2\binom{\floor{\sqrt{n}}}{3}}$ by applying $\arc01$ moves to the curve $T(\floor{\sqrt{n}}+1, \floor{\sqrt{n}})$.  Similarly, we obtain an $n$-vertex curve with defect $\smash{2\binom{\floor{\sqrt{n+1}}+1}{3}}$ by adding loops to the curve $T(\floor{\sqrt{n+1}}, \floor{\sqrt{n+1}}+1)$.  Lemma \ref{L:defect} now immediately implies the lower bound on homotopy moves.
\end{proof}

\subsection{Multicurves}
\label{SS:multi-lower}

Our previous results immediately imply that simplifying a multicurve with $n$ vertices requires at least $\Omega(n^{3/2})$ homotopy moves; in this section we derive additional lower bounds in terms of the number of constituent curves.  We distinguish between two natural variants of simplification: transforming a multicurve into an \emph{arbitrary} set of disjoint simple closed curves, or into a \emph{particular} set of disjoint simple closed curves.

Both lower bound proofs rely on the classical notion of \EMPH{winding number}.  Let $\gamma$ be an arbitrary closed curve in the plane, let $p$ be any point outside the image of $\gamma$, and let~$\rho$ be any ray from $p$ to infinity that intersects~$\gamma$ transversely.  The winding number of $\gamma$ around $p$, which we denote \EMPH{$\Wind(\gamma, p)$}, is the number of times~$\gamma$ crosses $\rho$ from right to left, minus the number of times $\gamma$ crosses $\rho$ from left to right.  The winding number does not depend on the particular choice of ray $\rho$.  All points in the same face of~$\gamma$ have the same winding number.  Moreover, if there is a homotopy from one curve $γ$ to another curve $γ’$, \EDIT{where the image of any intermediate curve} does not include $p$, then $\Wind(γ, p) = \Wind(γ', p)$ \cite{h-udtse-35}.

\begin{lemma}
Transforming a $k$-curve with $n$ vertices in the plane into $k$ arbitrary disjoint circles requires $\Omega(nk)$ homotopy moves in the worst case.
\end{lemma}

\begin{proof}
For arbitrary positive integers $n$ and $k$, we construct a multicurve with $k$ disjoint constituent \EDIT{curves}, all but one of which are simple, as follows.  The first $k-1$ constituent \EDIT{curves} $\gamma_1, \dots, \gamma_{k-1}$ are disjoint circles inside the open unit disk centered at the origin.  (The precise configuration of these circles is unimportant.)  The remaining \EDIT{constituent} curve $\gamma_o$ is a spiral winding $n+1$ times around the closed unit disk centered at the origin, plus a line segment connecting the endpoints of the spiral; $\gamma_o$ is the simplest possible curve with winding number $n+1$ around the origin. Let $\gamma$ be the disjoint union of these $k$ curves; we claim that $\Omega(nk)$ homotopy moves are required to simplify $\gamma$.  See Figure \ref{F:winding}.

\begin{figure}[ht]
\centering
\includegraphics[scale=0.5]{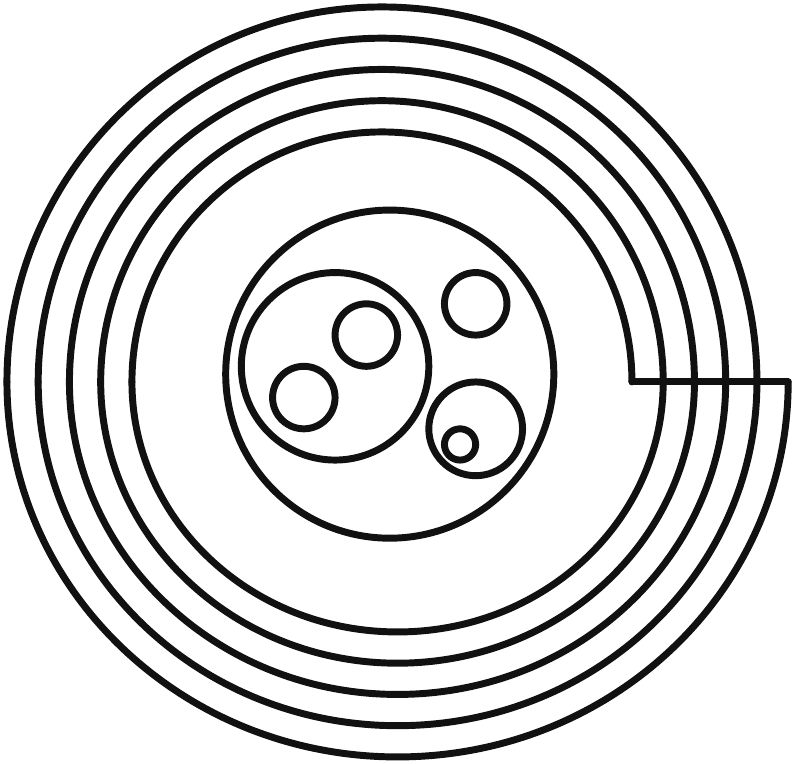}
\caption{Simplifying this multicurve requires $\Omega(nk)$ homotopy moves.}
\label{F:winding}
\end{figure}

Consider the faces of the outer curve $\gamma_o$ during any homotopy of $\gamma$.  Adjacent faces of $\gamma_o$ have winding numbers that differ by $1$, and the outer face has winding number $0$.  Thus, for any non-negative integer~$w$, as long as the maximum absolute winding number $\Abs{\max_p \Wind(\gamma_o,p)}$ is at least $w$, the curve $\gamma_o$ has at least $w+1$ faces (including the outer face) and therefore at least $w-1$ vertices, by Euler's formula.  On the other hand, if any curve~$\gamma_i$ intersects a face of $\gamma_o$, no homotopy move can remove that face \EDIT{until the intersection between $\gamma_i$ and $\gamma_o$ is removed}.  Thus, before the simplification of $\gamma_o$ is complete, each curve~$\gamma_i$ must intersect only faces with winding number $0$, $1$, or $-1$.

For each index $i$, let $w_i$ denote the maximum absolute winding number of $\gamma_o$ around any point of~$\gamma_i$:
\[
	w_i \coloneqq \max_\theta \Abs{\Wind\left(\gamma_o,\strut \gamma_i(\theta)\right)}.
\]
Let $W = \sum_i w_i$.  Initially, $W = k(n+1)$, and when $\gamma_o$ first becomes simple, we must have $W \le k$.  Each homotopy move changes $W$ by at most $1$; specifically, at most one term $w_i$ changes at all, and that term either increases or decreases by $1$.  The $\Omega(nk)$ lower bound now follows immediately.
\end{proof}

\begin{theorem}
\label{Th:multi-lower}
Transforming a $k$-curve with $n$ vertices in the plane into an arbitrary set of $k$ simple closed curves requires $\Omega(n^{3/2} + nk)$ homotopy moves in the worst case.
\end{theorem}

We say that a collection of $k$ disjoint simple closed curves is \EMPH{nested} if some point lies in the interior of every curve, and \EMPH{unnested} if the curves have disjoint interiors.

\begin{lemma}
Transforming $k$ nested circles in the plane into $k$ unnested circles requires $\Omega(k^2)$ homotopy moves.
\end{lemma}

\begin{proof}
Let $\gamma$ and $\gamma'$ be two nested circles, with $\gamma'$ in the interior of~$\gamma$ and with $\gamma$ directed counterclockwise.  Suppose we apply an arbitrary homotopy to these two curves.  If the curves remain disjoint during the entire homotopy, then $\gamma'$ always lies inside a face of $\gamma$ with winding number~$1$; in short, the two curves remain nested.
%
%
Thus, any sequence of homotopy moves that takes $\gamma$ and $\gamma'$ to two non-nested simple closed curves contains at least one $\arc{0}{2}$ move that makes the curves cross (and symmetrically at least one $\arc{2}{0}$ move that makes them disjoint again).

\begin{figure}[ht]
\centering
\includegraphics[scale=0.5]{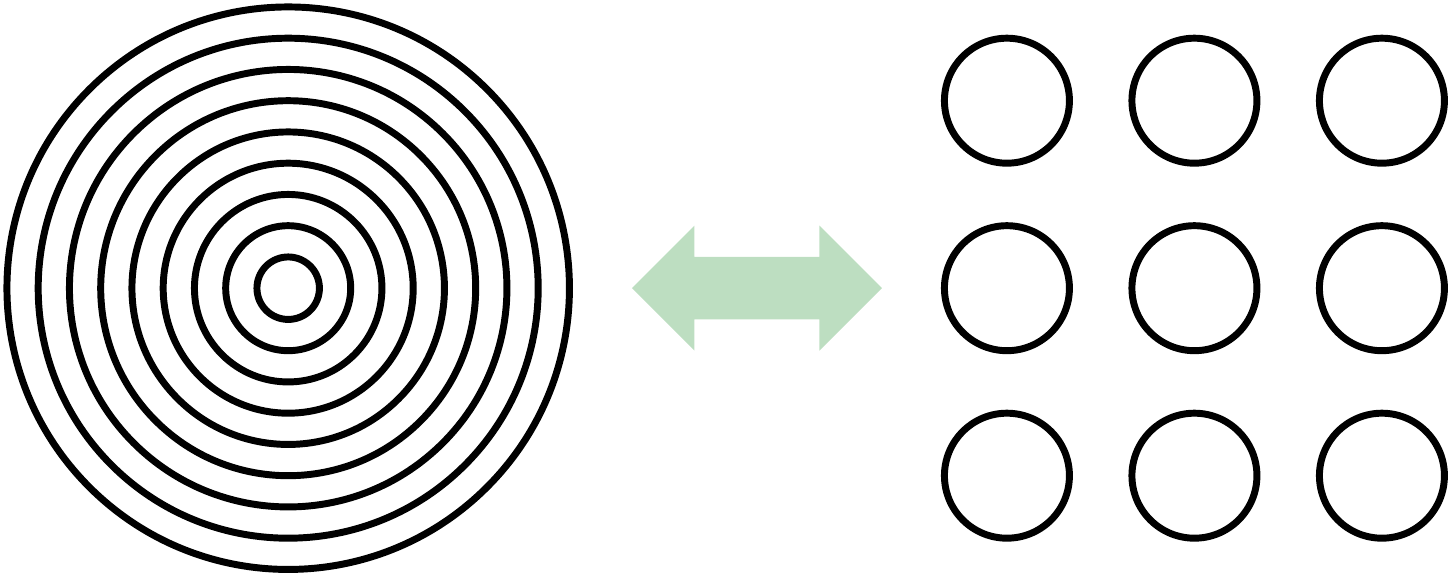}
\caption{Nesting or unnesting $k$ circles requires $\Omega(k^2)$ homotopy moves.}
\end{figure}

Consider a set of $k$ nested circles.  Each of the $\smash{\binom{k}{2}}$ pairs of circles requires at least one $\arc{0}{2}$ move and one $\arc{2}{0}$ move to unnest.  Because these moves involve distinct pairs of curves, at least $\smash{\binom{k}{2}}$ $\arc{0}{2}$ moves and $\smash{\binom{k}{2}}$ $\arc{2}{0}$ moves, and thus at least $k^2-k$ moves altogether, are required to unnest every pair.
\end{proof}

\begin{theorem}
\label{Th:multi-lower2}
Transforming a $k$-curve with $n$ vertices in the plane into $k$ nested (or unnested) circles requires $\Omega(n^{3/2} + nk + k^2)$ homotopy moves in the worst case.
\end{theorem}

\begin{corollary}
\label{C:multi-lower3}
Transforming one $k$-curve with at most $n$ vertices into another $k$-curve with at most $n$ vertices requires  $\Omega(n^{3/2} + nk + k^2)$ homotopy moves in the worst case.
\end{corollary}

Although our lower bound examples consist of disjoint curves, all of these lower bounds apply without modification to \emph{connected} multicurves, because any $k$-curve can be connected with at most $k-1$ $\arc{0}{2}$ moves.  On the other hand, any connected $k$-curve has at least $2k-2$ vertices, so the $\Omega(k^2)$ terms in Theorem~\ref{Th:multi-lower2} and Corollary \ref{C:multi-lower3} are redundant.

\section{Electrical Transformations}
\label{S:electric}

Now we consider a related set of local operations on plane graphs, called \EMPH{\EDIT{facial} electrical transformations}, consisting of six operations in three dual pairs, as shown in Figure \ref{F:elec-dual}.
\begin{itemize}\itemsep0pt
\item
\emph{degree-$1$ reduction}: Contract the edge incident to a vertex of degree $1$, or delete the edge incident to a face of degree $1$
\item
\emph{series-parallel reduction}: Contract either edge incident to a vertex of degree $2$, or delete either edge incident to a face of degree $2$
\item
\emph{$\Delta Y$ transformation}: Delete a vertex of degree 3 and connect its neighbors with three new edges, or delete the edges bounding a face of degree 3 and join the vertices of that face to a new vertex.
\end{itemize}

\begin{figure}[ht]
\centering
\includegraphics[width=0.85\textwidth]{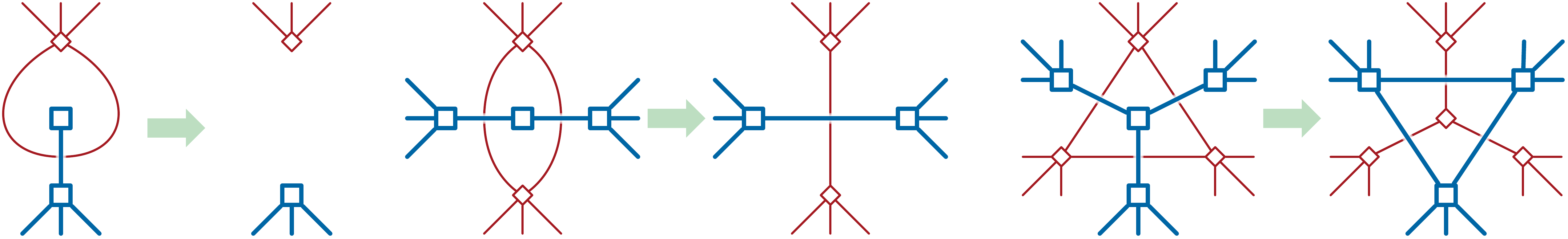}
\caption{Facial electrical transformations in a plane graph $G$ and its dual graph $G^*$.}
\label{F:elec-dual}
\end{figure}

\EDIT{Electrical transformations are usually defined more generally as a set of operations performed on abstract graphs, which} have been used since the end of the 19th century~\cite{k-etscn-1899,r-md-1904} to analyze resistor networks and other electrical circuits, but have since been applied to a number of other combinatorial problems on planar graphs, including shortest paths and maximum flows~\cite{a-wtns-60}; multicommodity flows~\cite{f-erpns-85}; and counting spanning trees, perfect matchings, and cuts~\cite{cpv-nastc-95}. We refer to our earlier preprint \cite[Section~1.1]{defect} for a more detailed history and an expanded list of applications.  \EDIT{However, all the algorithms we describe below reduce any plane graph to a single vertex using only \emph{facial} electrical transformations as defined above.}

In light of these applications, it is natural to ask \emph{how many} \EDIT{facial} electrical transformations are required in the worst case. 
%
\EDIT{The earliest algorithm for reducing a plane graph to a single vertex already follows from Steinitz's bigon-reduction argument, which we described in the introduction~\cite{s-pr-1916,sr-vtp-34}.  Steinitz reduced local transformations of \EDIT{plane} \emph{graphs} to local transformations of planar \emph{curves} by defining the \emph{medial graphs} (“$\Theta$-Prozess”), which we consider in detail below.  Later algorithms were given by Truemper \cite{t-drpg-89}, Feo and Provan \cite{fp-dtert-93}, and others.  Both Steinitz’s algorithm and Feo and Provan’s algorithm require at most $O(n^2)$ facial electrical transformations; this is the best upper bound known.}

Even the special case of regular grids is open and interesting. Truemper~\cite{t-drpg-89,t-md-92} describes a method to reduce the $p\times p$ grid in $O(p^3)$ steps.  Nakahara and Takahashi~\cite{nt-aafts-96} prove an upper bound of $O(\min\set{pq^2, p^2q})$ for the $p\times q$ cylindrical grid.  Because every $n$-vertex \EDIT{plane} graph is a minor of an $O(n)\times O(n)$ grid~\cite{v-ucvc-81,s-mncpe-84}, both of these results imply an $O(n^3)$ upper bound for arbitrary plane graphs; see Lemma \ref{L:smoothing}.  Feo and Provan~\cite{fp-dtert-93} claim without proof that Truemper's algorithm actually performs only $O(n^2)$ electrical transformations.  On the other hand, the smallest (cylindrical) grid containing every $n$-vertex plane graph as a minor has size $Ω(n) \times Ω(n)$~\cite{v-ucvc-81}.  Archdeacon \etal~\cite{acgp-frpwg-00} asked whether the $O(n^{3/2})$ upper bound for square grids can be improved to near-linear:
\begin{quote}\small
It is possible that a careful implementation and analysis of the grid-embedding schemes can lead to an $O(n\sqrt{n})$-time algorithm for the general planar case. It would be interesting to obtain a near-linear algorithm for the grid\dots. However, it may well be that reducing planar grids is $Ω(n\sqrt{n})$.
\end{quote}

\EDIT{Most of these earlier algorithms actually solve a more difficult problem, first considered by Akers~\cite{a-wtns-60} and Lehman~\cite{l-wtpn-63} and later solved by Epifanov \cite{e-rpges-66}, of reducing a planar graph with two special vertices called \emph{terminals} to a single edge between the two terminals.  In this context, any electrical transformation that contracts an edge incident to a terminal is forbidden.  Unfortunately, not every two-terminal plane graph can be reduced to a single edge using only facial electrical transformations; Figure~\ref{F:bullseye} shows two examples.  However, it is sufficient to allow loop reductions, parallel reductions, and $\arc\Delta Y$ transformations to be performed on faces that contain a terminal vertex of degree $1$ (and nothing else) \cite{fp-dtert-93}.  It is an open question whether our lower bound still holds if these additional non-facial transformations are allowed.}

\begin{figure}[ht]
\centering
\includegraphics[scale=0.3]{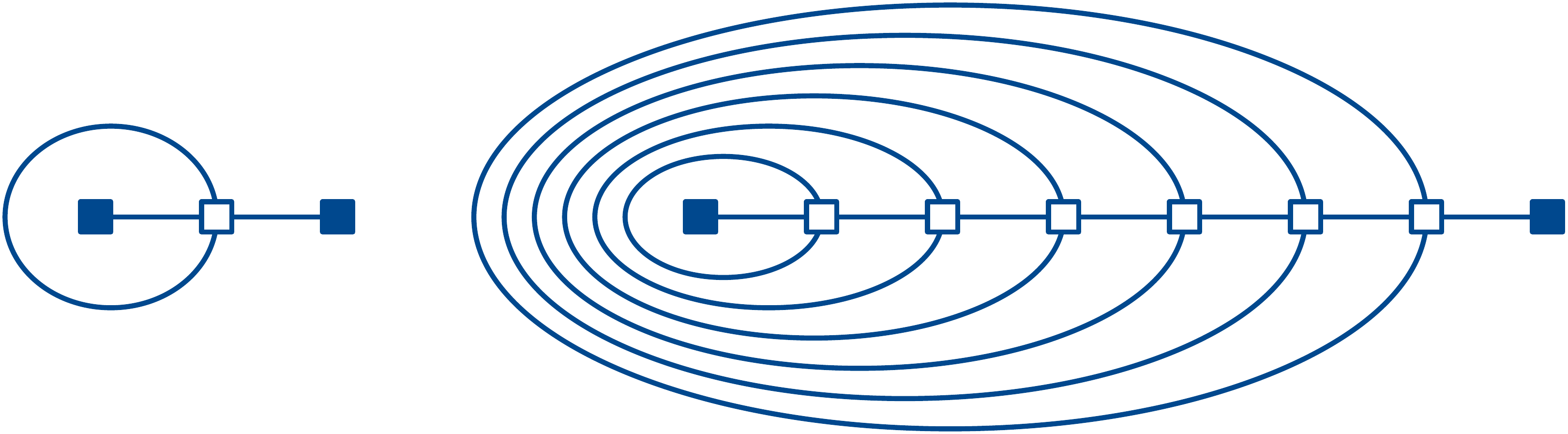}
\caption{Plane graphs with two terminals that cannot be further reduced using only facial electrical transformations.}
\label{F:bullseye}
\end{figure}

\subsection{Definitions}

The \EMPH{medial graph} of a plane graph $G$, which we denote \EMPH{$G^\times$}, is another plane graph whose vertices correspond to the edges of $G$ and whose edges correspond to incidences (with multiplicity) between vertices of $G$ and faces of~$G$. Two vertices of $G^\times$ are connected by an edge if and only if the corresponding edges in~$G$ are consecutive in cyclic order around some vertex, or equivalently, around some face in~$G$.  Every vertex in every medial graph has degree $4$; thus, every medial graph is the image of a multicurve.  The medial graphs of any plane graph~$G$ and its dual $G^*$ are identical.  
To avoid trivial boundary cases, we define the medial graph of an isolated vertex to be a circle.

\EDIT{Facial} electrical transformations in any plane graph $G$ correspond to local transformations in the medial graph $G^\times$ that are almost identical to homotopy moves. Each degree-$1$ reduction in $G$ corresponds to a $\arc10$ homotopy move in $G^\times$, and each $\Delta$Y transformation in~$G$ corresponds to a $\arc33$ homotopy move in~$G^\times$. A series-parallel reduction in $G$ contracts an empty bigon in $G^\times$ to a single vertex. Extending our earlier notation, we call this transformation a \EMPH{$\arc21$} move. We collectively refer to these transformations and their inverses as \EMPH{medial electrical moves}; see Figure~\ref{F:medial-elec}.

\begin{figure}[ht]
\centering
\includegraphics[width=0.9\linewidth]{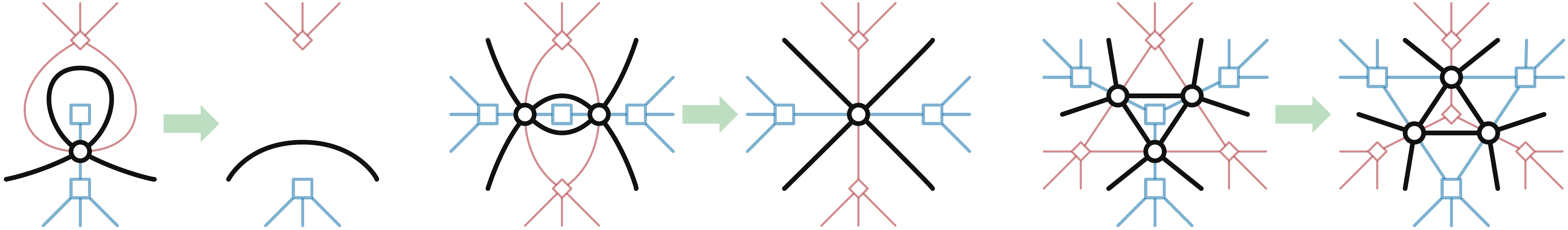}\\
\caption{Medial electrical moves $\arc10$, $\arc21$, and $\arc33$.}
\label{F:medial-elec}
\end{figure}

\EMPH{Smoothing} a \EDIT{multicurve} $γ$ at a vertex $x$ means replacing the intersection of $γ$ with a small neighborhood of $x$ with two disjoint simple paths, so that the result is another \EDIT{multicurve}. (There are two possible smoothings at each vertex; see Figure \ref{F:smoothing}.)  A \EMPH{smoothing} of $γ$ is any graph obtained by smoothing zero or more vertices of $γ$, and a \EMPH{proper smoothing} of $γ$ is any smoothing other than $\gamma$ itself. For any plane graph $G$, the (proper) smoothings of the medial graph $G^\times$ are precisely the medial graphs of (proper) minors of $G$.

\begin{figure}[ht]
\centering
\includegraphics[scale=0.3]{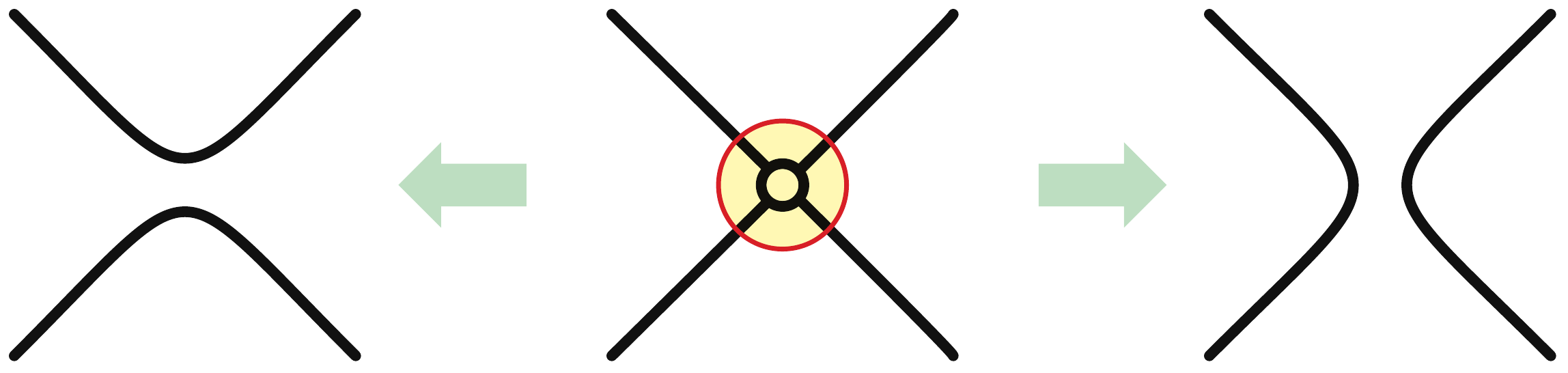}
\caption{Smoothing a vertex.}
\label{F:smoothing}
\end{figure}

\subsection{Electrical to Homotopy}

The main result of this section is that the number of \emph{homotopy} moves required to simplify a closed curve is a lower bound on the number of \emph{medial electrical moves} required to simplify the same closed curve.  This result is already implicit in the work of Noble and Welsh~\cite{nw-kg-00}, and most of our proofs closely follow theirs.  We include the proofs here to make the inequalities explicit and to keep the paper self-contained.

For any connected multicurve (or 4-regular \EDIT{plane} graph) $\gamma$, let \EMPH{$X(γ)$} denote the minimum number of medial electrical moves required to reduce $γ$ to a simple closed curve, and let \EMPH{$H(γ)$} is the minimum number of homotopy moves required to reduce $γ$ to an arbitrary collection of disjoint simple closed curves.

The following key lemma follows from close reading of proofs by Truemper~\cite[Lemma~4]{t-drpg-89} and several others~\cite{g-dtaa-91,nt-aafts-96,acgp-frpwg-00,nw-kg-00} that every minor of a ΔY-reducible graph is also ΔY-reducible.  Our proof most closely resembles an argument of Gitler~\cite[Lemma~2.3.3]{g-dtaa-91}, but restated in terms of medial electrical moves to simplify the case analysis.


  
\begin{lemma}
\label{L:smoothing}
For any connected plane graph $G$, reducing any connected proper minor of $G$ to a single vertex requires strictly fewer \EDIT{facial} electrical transformations than reducing $G$ to a single vertex.
Equivalently, $X(\overline{γ}) < X(γ)$ for every connected proper smoothing $\overline{γ}$ of every connected multicurve $γ$.
\end{lemma}

\begin{proof}
Let $γ$ be a connected multicurve, and let $\overline{γ}$ be a connected proper smoothing of $γ$.  If $γ$ is already simple, the lemma is vacuously true.  Otherwise, the proof proceeds by induction on $X(γ)$.

We first consider the special case where $\overline{γ}$ is obtained from $γ$ by smoothing a single vertex~$x$.  Let~$γ'$ be the result of the first medial electrical move in the minimum-length sequence that reduces $γ$ \EDIT{to a simple closed curve}.  We immediately have $X(γ) = X(γ')+1$.  There are two nontrivial cases to consider.

First, suppose the move from $γ$ to $γ'$ does not involve the smoothed vertex $x$.  Then we can apply the same move to $\overline{γ}$ to obtain a new graph $\overline{γ}'$; the same graph  can also be obtained from $γ'$ by smoothing~$x$.  We immediately have $X(\overline{γ}) \le X(\overline{γ}') + 1$, and the inductive hypothesis implies $X(\overline{γ}')+1 < X(γ')+1 = X(γ)$.

Now suppose the first move in $Σ$ does involve $x$. In this case, we can apply at most one medial electrical move to $\overline{γ}$ to obtain a (possibly trivial) smoothing $\overline{γ}'$ of $γ'$.  There are eight subcases to consider, shown in Figure \ref{F:smooth-moves}.  One subcase for the $\arc01$ move is impossible, because~$\overline{γ}$ is connected.  In the remaining $\arc01$ subcase and one $\arc21$ subcase, the curves $\overline{γ}$, $\overline{γ}'$ and $γ'$ are all isomorphic, which implies $X(\overline{γ}) = X(\overline{γ}') = X(γ') = X(γ)-1$.  In all remaining subcases, $\overline{γ}'$ is a connected proper smoothing of $γ'$, so the inductive hypothesis implies $X(\overline{γ}) ≤ X(\overline{γ}')+1 < X(γ')+1 = X(γ)$.

\begin{figure}[ht]
\centering
\includegraphics[width=0.8\textwidth]{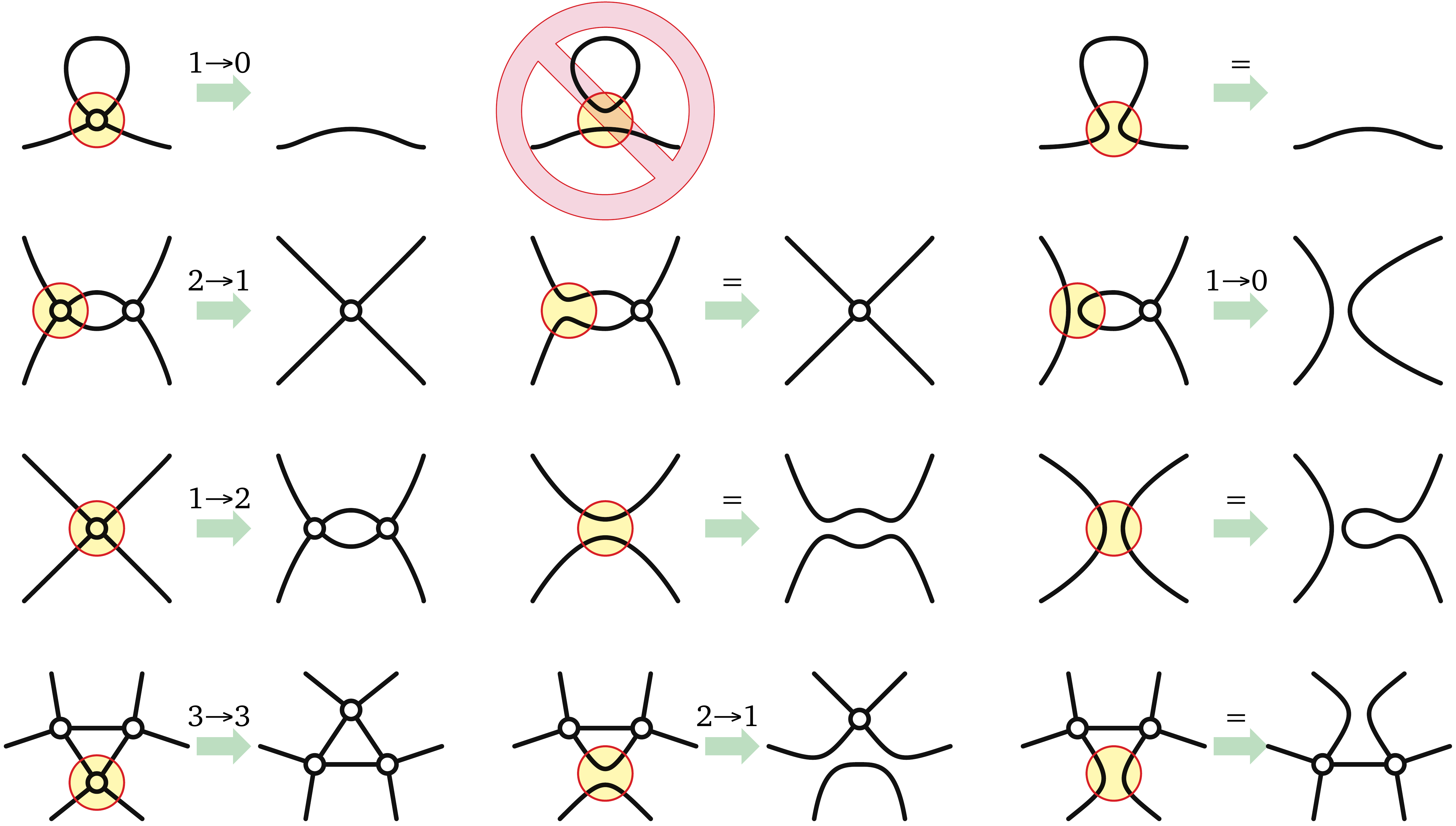}
\caption{Cases for the proof of the Lemma~\ref{L:smoothing}; the circled vertex is $x$.}
\label{F:smooth-moves}
\end{figure}

Finally, we consider the more general case where $\overline{γ}$ is obtained from $γ$ by smoothing more than one vertex.  Let $\widetilde{γ}$ be any intermediate curve, obtained from~$γ$ by smoothing just one of the vertices that were smoothed to obtain $\overline{γ}$.  As $\overline{γ}$ is a connected smoothing of $\widetilde{γ}$, the curve $\widetilde{γ}$ itself must be connected too.
Our earlier argument implies that $X(\widetilde{γ}) < X(γ)$.  Thus, the inductive hypothesis implies $X(\overline{γ}) < X(\widetilde{γ})$, which completes the proof.
\end{proof}

\begin{lemma}
\label{L:monotonicity}
For every connected multicurve $γ$, there is a minimum-length sequence of medial electrical moves that reduces $γ$ \EDIT{to a simple closed curve} and that does not contain $\arc01$ or $\arc12$ moves. 
\end{lemma}

\begin{proof}
Our proof follows an argument of Noble and Welsh~\cite[Lemma~3.2]{nw-kg-00}.

Consider a minimum-length sequence of medial electrical moves that reduces an arbitrary connected multicurve 	$γ$ \EDIT{to a simple closed curve}. For any integer $i ≥ 0$, let $γ_i$ denote the result of the first $i$ moves in this sequence; in particular, $γ_0 = γ$ and $γ_{X(γ)}$ is a \EDIT{simple closed curve}. Minimality of the reduction sequence implies that $X(γ_i) = X(γ)-i$ for all $i$. Now let $i$ be an arbitrary index such that $γ_i$ has one more vertex than $γ_{i-1}$. Then $γ_{i-1}$ is a connected proper smoothing of $γ_i$, so Lemma~\ref{L:smoothing} implies that $X(γ_{i-1}) < X(γ_i)$, giving us a contradiction. 
\end{proof}

\begin{lemma}
\label{L:homotopy}
$X(γ) ≥ H(γ)$ for every closed curve $γ$. 
\end{lemma} 

\begin{proof}
The proof proceeds by induction on $X(γ)$, following an argument of Noble and Welsh~\cite[Proposition 3.3]{nw-kg-00}.

Let $γ$ be a closed curve.  If $X(γ) = 0$, then $γ$ is already simple, so $H(γ) = 0$.  Otherwise, let $Σ$ be a minimum-length sequence of medial electrical moves that reduces $γ$ to a \EDIT{simple closed curve}.  Lemma~\ref{L:monotonicity} implies that we can assume that the first move in $Σ$ is neither $\arc01$ nor $\arc12$. If the first move is $1\arcto 0$ or $3\arcto 3$, the theorem immediately follows by induction. 

The only interesting first move is $2\arcto 1$. Let $γ'$ be the result of this $\arc21$ move, and let $\overline{γ}$ be the result of the corresponding $\arc20$ homotopy move.  The minimality of $Σ$ implies that $X(γ) = X(γ')+1$, and we trivially have $H(γ) \le H(\overline{γ}) + 1$.  Because $γ$ consists of \emph{one} single curve, $\overline{γ}$ \EDIT{is also a single curve and is therefore} connected.  The curve $\overline{γ}$ is also a proper smoothing of $γ'$, so the Lemma~\ref{L:smoothing} implies $X(\overline{γ}) < X(γ') < X(γ)$.  Finally, the inductive hypothesis implies that $X(\overline{γ}) \ge H(\overline{γ})$, which completes the proof.
\end{proof}

We call a plane graph $G$ \EMPH{unicursal} if its medial graph $G^\times$ is the image of a single closed curve.

\begin{theorem}
\label{Th:lowerbound}
For every connected \EDIT{plane} graph $G$ and every unicursal minor~$H$ of $G$, reducing $G$ to a single vertex requires at least $\abs{\Defect(H^\times)}/2$ \EDIT{facial} electrical transformations. 
\end{theorem}

\begin{proof}
Either $H$ equals $G$, or Lemma~\ref{L:smoothing} implies that reducing a proper minor $H$ of $G$ to a single vertex requires strictly fewer \EDIT{facial} electrical transformations than reducing $G$ to a single vertex.  Note that \EDIT{facial} electrical transformations performed on $H$ corresponds precisely to medial electrical moves performed on $H^\times$.  Now because $γ \coloneqq H^\times$ is unicursal, Lemma~\ref{L:defect} and Lemma~\ref{L:homotopy} implies that $X(γ) ≥ H(γ) ≥ \abs{\Defect(\gamma)}/2$.
\end{proof}

\subsection{Cylindrical Grids}

Finally, we derive explicit lower bounds for the number of \EDIT{facial} electrical transformations required to reduce any cylindrical grid to a single vertex.  For any positive integers $p$ and $q$, we define two cylindrical grid graphs; see Figure \ref{F:cylinders}.
\begin{itemize}
\item
\EMPH{$C(p,q)$} is the Cartesian product of a cycle of length $q$ and a path of length $p-1$.  If $q$ is odd, then the medial graph of $C(p,q)$ is the flat torus knot $T(2p, q)$.

\item
\EMPH{$C'(p,q)$} is obtained by connecting a new vertex to the vertices of one of the $q$-gonal faces of $C(p,q)$, or equivalently, by contracting one of the $q$-gonal faces of $C(p+1,q)$ to a single vertex.  If~$q$ is even, then the medial graph of $C'(p,q)$ is the flat torus knot $T(2p+1, q)$.
\end{itemize}
\unskip

\begin{figure}[ht]
\centering
\hfil
\includegraphics[width=1.75in]{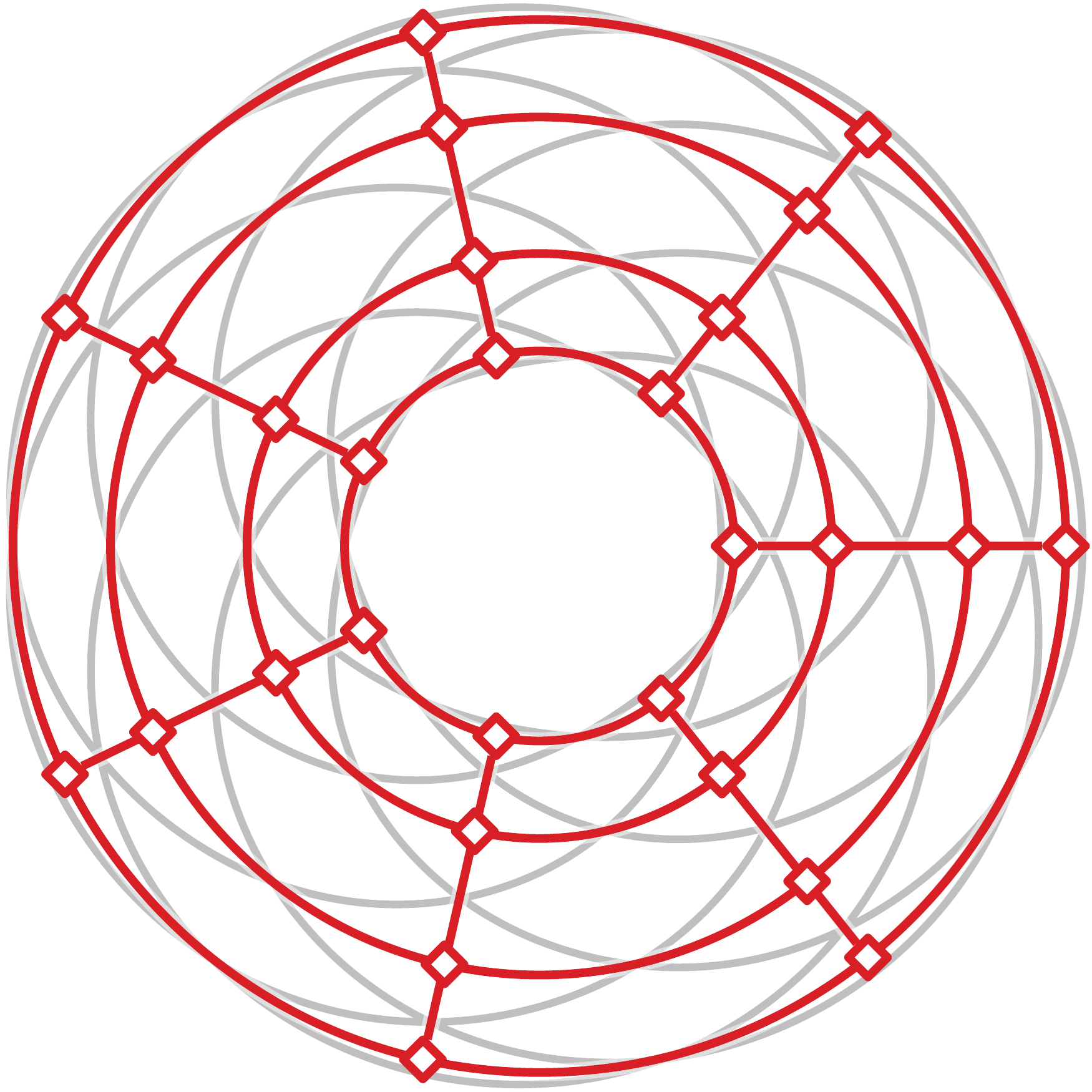}\hfil
\includegraphics[width=1.75in]{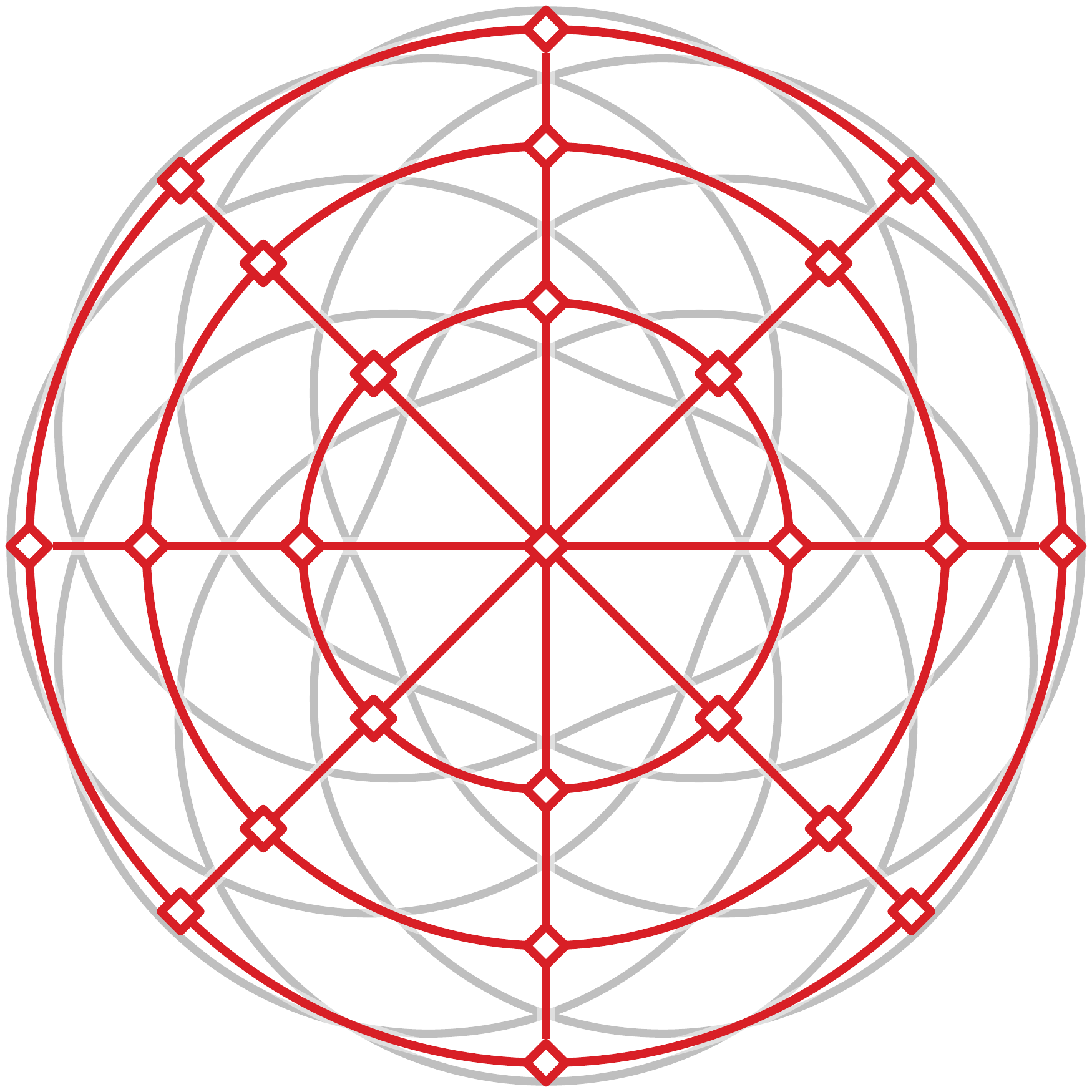}\hfil{}
\caption{The cylindrical grid graphs $C(4,7)$ and $C'(3,8)$ and (in light gray) their medial graphs $T(8,7)$ and $T(7,8)$.}
\label{F:cylinders}
\end{figure}

\begin{corollary}
\label{C:cylindrical-grid}
For all positive integers $p$ and $q$, the cylindrical grid $C(p,q)$ requires $Ω(\min\set{p^2 q, p q^2})$ \EDIT{facial} electrical transformations to reduce to a single vertex.
\end{corollary}

\begin{proof}
First suppose $p \le q$.  Because $C(p-1,q)$ is a minor of $C(p,q)$, we can assume without loss of generality that $p$ is even and $p<q$.  Let $H$ denote the cylindrical grid $C(p/2, ap+1)$, where $a \coloneqq \floor{(q-1)/p} \ge 1$.  $H$~is a minor of $C(p, q)$ (because $ap+1 \le q$), and the medial graph of $H$ is the flat torus knot $T(p, ap+1)$.  Lemma~\ref{L:braid-wide} implies
\[
	\Defect\!\left(T(p,\strut ap+1)\right) = 2a\binom{p+1}{3} = Ω(ap^3) = Ω(p^2q).
\]
Theorem~\ref{Th:lowerbound} now implies that reducing $C(p,q)$ requires at least $\Omega(p^2 q)$ \EDIT{facial} electrical transformations.

The symmetric case $p > q$ is similar.  We can assume without loss of generality that $q$ is odd.  Let $H$ denote the cylindrical grid $C'(aq, q)$, where $a \coloneqq \floor{(p-1)/q} \ge 1$.  $H$ is a proper minor of $C(p, q)$ (because $aq<p$), and the medial graph of $H$ is the flat torus knot $T(2aq+1, q)$.  Lemma~\ref{L:braid-deep} implies
\[
	\Abs{ \Defect\!\left(T(2aq+1,\strut q)\right)}
		= 4a\binom{q}{3} = Ω(aq^3) = Ω(pq^2).
\]
Theorem~\ref{Th:lowerbound} now implies that reducing $C(p,q)$ requires at least $\Omega(p q^2)$ \EDIT{facial} electrical transformations.
\end{proof}

In particular, reducing any $\Theta(\sqrt{n})\times\Theta(\sqrt{n})$ cylindrical grid requires at least $\Omega(n^{3/2})$ \EDIT{facial} electrical transformations.  Our lower bound matches an $O(\min\set{pq^2, p^2q})$ upper bound by Nakahara and Takahashi~\cite{nt-aafts-96}.  Because every $p\times q$ rectangular grid contains $C(\floor{p/3}, \floor{q/3})$ as a minor, the same $Ω(\min\set{p^2 q, p q^2})$ lower bound applies to rectangular grids.  In particular, Truemper's $O(p^3) = O(n^{3/2})$ upper bound for the $p\times p$ square grid~\cite{t-drpg-89} is tight.  Finally, because every \EDIT{plane} graph with treewidth~$t$ contains an $Ω(t)\times Ω(t)$ grid minor~\cite{rst-qepg-94}, reducing \emph{any} $n$-vertex \EDIT{plane} graph with treewidth~$t$ requires at least $Ω(t^3 + n)$ \EDIT{facial} electrical transformations. 

Like Gitler~\cite{g-dtaa-91}, Feo and Provan~\cite{fp-dtert-93}, and Archdeacon \etal~\cite{acgp-frpwg-00}, we conjecture that any $n$-vertex \EDIT{plane} graph can be reduced \EDIT{to a vertex} using $O(n^{3/2})$ \EDIT{facial} electrical transformations.  More ambitiously, we conjecture an upper bound of $O(nt)$ for any $n$-vertex \EDIT{plane} graph with treewidth $t$.

\section{Upper Bound}
\label{S:upper}

For any point $p$, let \EMPH{$\Depth(p, \gamma)$} denote the minimum number of times a path from~$p$ to infinity crosses~$\gamma$. Any two points in the same face of $\gamma$ have the same depth, so each face~$f$ has a well-defined depth, which is its distance to the outer face in the dual graph of~$\gamma$; see Figure~\ref{F:curve-depth-contract}. The depth of the curve, denoted \EMPH{$\Depth(\gamma)$}, is the maximum depth of the faces of~$\gamma$; and the \EMPH{potential $D(\gamma)$} is the sum of the depths of the faces.  Euler's formula implies that any 4-regular \EDIT{plane} graph with $n$ vertices has exactly $n+2$ faces; thus, for any curve $\gamma$ with $n$ vertices, we have $n+1 \le D(\gamma) \le (n+1)\cdot \Depth(\gamma)$.

\subsection{Contracting Simple Loops}

\begin{lemma}
\label{L:contract}
Every closed curve $\gamma$ in the plane can be simplified using at most $3D(\gamma) - 3$ homotopy moves.
\end{lemma}

\begin{proof}
\EDIT{We prove the statement by induction on the number of vertices in $\gamma$.}
The lemma is trivial if $\gamma$ is already simple, so assume otherwise.
Let $x \coloneqq \gamma(\theta) = \gamma(\theta')$ be the first vertex to be visited twice by $\gamma$ after the (arbitrarily chosen) basepoint $\gamma(0)$.
%
%
Let~$\alpha$ denote the subcurve of~$\gamma$ from $\gamma(\theta)$ to $\gamma(\theta')$; our choice of $x$ implies that $\alpha$ is a simple loop. Let $m$ and $s$ denote the number of vertices and maximal subpaths of $\gamma$ in the interior of $\alpha$ respectively.

Finally, let $\gamma'$ denote the closed curve obtained from $\gamma$ by removing $\alpha$. The first stage of our algorithm transforms~$\gamma$ into $\gamma'$ by contracting the loop $\alpha$ via homotopy moves.  

\begin{figure}[ht]
\centering\includegraphics[scale=0.5]{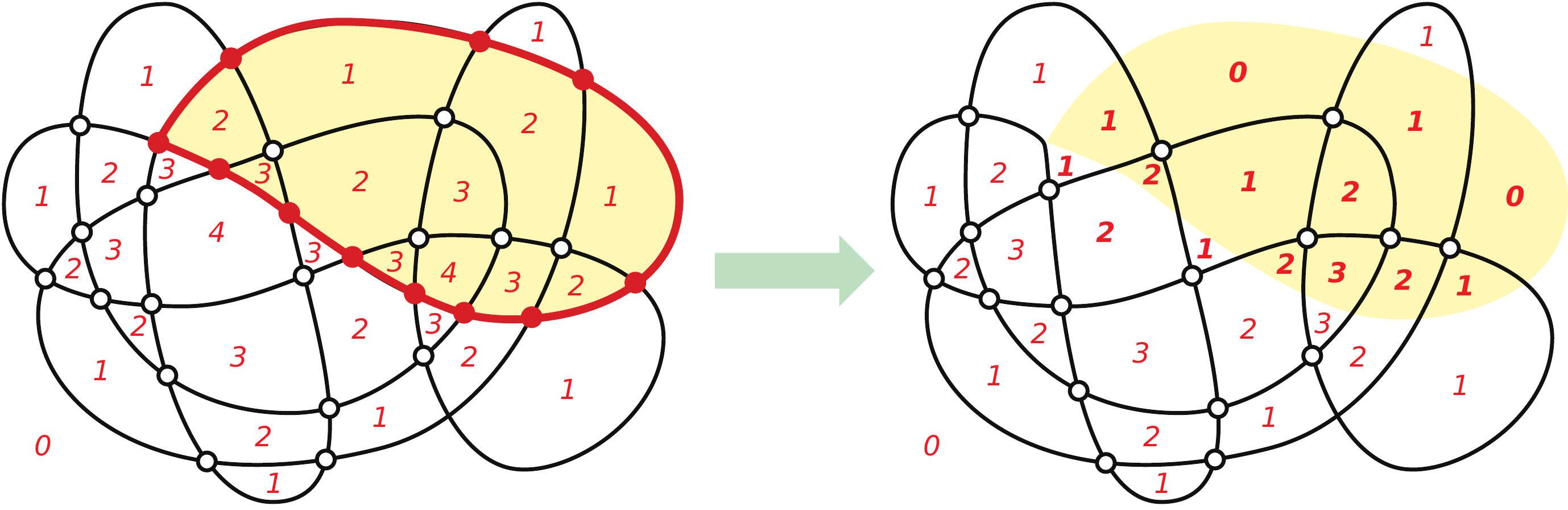}
\caption{Transforming $\gamma$ into $\gamma'$ by contracting a simple loop. Numbers are face depths.}
\label{F:curve-depth-contract}
\end{figure}

We remove the vertices and edges from the interior of $\alpha$ one at a time as follows \EDIT{(see Figure \ref{F:move-strand})}.
If we can perform a $\arc20$ move to remove one edge of $\gamma$ from the interior of $\alpha$ and decrease $s$, we do so.  Otherwise, either $\alpha$ is empty, or some vertex of $\gamma$ lies inside $\alpha$.  In the latter case, at least one vertex $x$  inside~$\alpha$ has a neighbor that lies on $\alpha$.  We move $x$ outside $\alpha$ with a $\arc02$ move (which increases $s$ \EDIT{by $1$}) followed by a $\arc33$ move \EDIT{(which decreases $m$ by $1$)}. 
Once $\alpha$ is an empty loop, we remove it with a single $\arc10$ move. Altogether, our algorithm transforms~$\gamma$ into~$\gamma'$ using at most $3m + s + 1$ homotopy moves.  Let $M$ denote the actual number of homotopy moves used.

\begin{figure}[ht]
\centering
\includegraphics[scale=0.5]{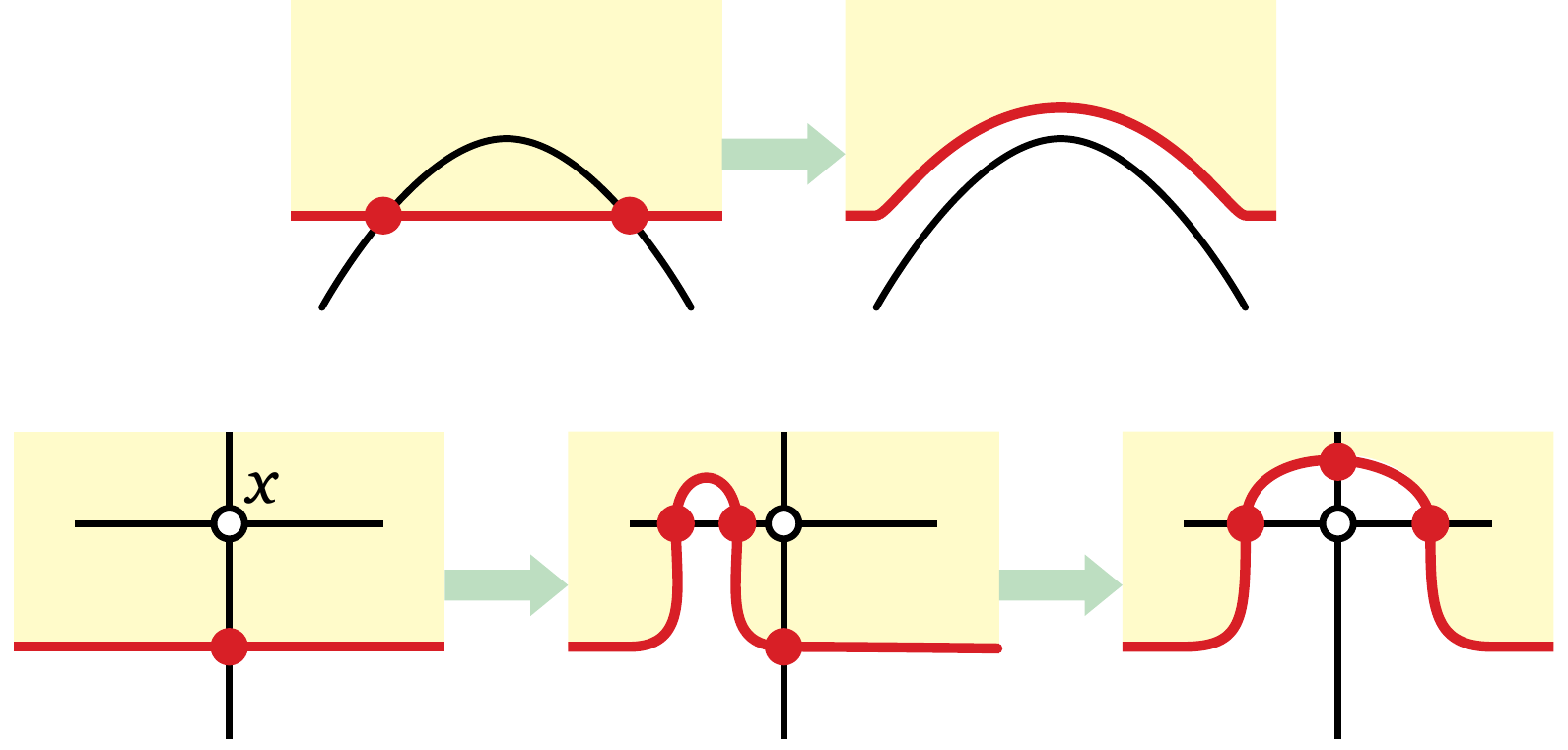}
\caption{Moving a loop over an interior empty bigon or an interior vertex. }
\label{F:move-strand}
\end{figure}

Euler's formula implies that $\alpha$ contains exactly \EDIT{$m+s+1$} faces of $\gamma$.  The Jordan curve theorem implies that $\Depth(p, \gamma') \le \Depth(p, \gamma)-1$ for any point $p$ inside $\alpha$, and trivially $\Depth(p, \gamma') \le \Depth(p, \gamma)$ for any point $p$ outside $\alpha$. It follows that $D(\gamma') \le D(\gamma) - (\EDIT{m+s+1})  \le D(\gamma) - M/3$, and therefore $M \le 3D(\gamma) - 3 D(\gamma')$. The induction hypothesis implies that we can recursively simplify $\gamma'$ using at most $3D(\gamma') - 3$ moves. The lemma now follows immediately.
\end{proof}

Our upper bound is a factor of 3 larger than Feo and Provan's \cite{fp-dtert-93}; however our algorithm has the advantage that it extends to \emph{tangles}, as described in the next subsection.

\subsection{Tangles}
\label{SS:tangles}

A \EMPH{tangle} is a collection of boundary-to-boundary paths $\gamma_1, \gamma_2, \dots,\gamma_s$ in a closed topological disk~$\Sigma$, which (self-)intersect only pairwise, transversely, and away from the boundary of~$\Sigma$.  This terminology is borrowed from knot theory, where a tangle usually refers to the intersection of a knot or link with a closed 3-dimensional ball \cite{c-eklst-70,cdm-ivki-12}; our tangles are perhaps more properly called \emph{flat tangles}, as they are  images of tangles under  appropriate projection.  (Our tangles are unrelated to the obstructions to small branchwidth introduced by Robertson and Seymour \cite{rs-gm10-91}.)  Transforming a curve into a tangle is identical to (an inversion of) the \emph{flarb} operation defined by Allen \etal~\cite{abil-ivd-16}.

We call each individual path $\gamma_i$ a \EMPH{strand} of the tangle. The \EMPH{boundary} of a tangle is the boundary of the disk $\Sigma$ that contains it; we usually denote the boundary by $\sigma$.  By the Jordan-Schönflies theorem, we can assume without loss of generality that $\sigma$ is actually a circle.  We can obtain a tangle from any closed curve $\gamma$ by considering its restriction to any closed disk whose boundary~$\sigma$ intersects $\gamma$ transversely away from its vertices; we call this restriction the \EMPH{interior tangle} of $\sigma$. 

The strands and boundary of any tangle define a plane graph $T$ whose boundary vertices each have degree $3$ and whose interior vertices each have degree $4$.  Depths and potential \EDIT{of a tangle} are defined exactly as for closed curves: The depth of any face $f$ of $T$ is its distance to the outer face in the dual graph $T^*$; the depth of the tangle is its maximum face depth; and the potential $D(T)$ of the tangle is the sum of its face depths.

A tangle is \EMPH{tight} if every pair of strands intersects at most once and \EMPH{loose} otherwise. Every loose tangle contains either an empty loop or a (not necessarily empty) bigon. Thus, any tangle with $n$ vertices can be transformed into a tight tangle---or less formally, \emph{tightened}---in $O(n^2)$ homotopy moves using Steinitz's algorithm.  On the other hand, there are infinite classes of loose tangles for which no homotopy move \EDIT{that} decreases the potential, so we cannot directly apply Feo and Provan's algorithm to this setting.

\EDIT{We describe a two-phase algorithm to tighten any tangle.  First, we remove any self-intersections in the individual strands, by contracting loops as in the proof of Lemma \ref{L:contract}.  Once each strand is simple, we move the strands so that each pair intersects at most once.  See Figure \ref{F:tighten-tangle}.}

\begin{figure}[ht]
\centering\includegraphics[width=0.9\linewidth]{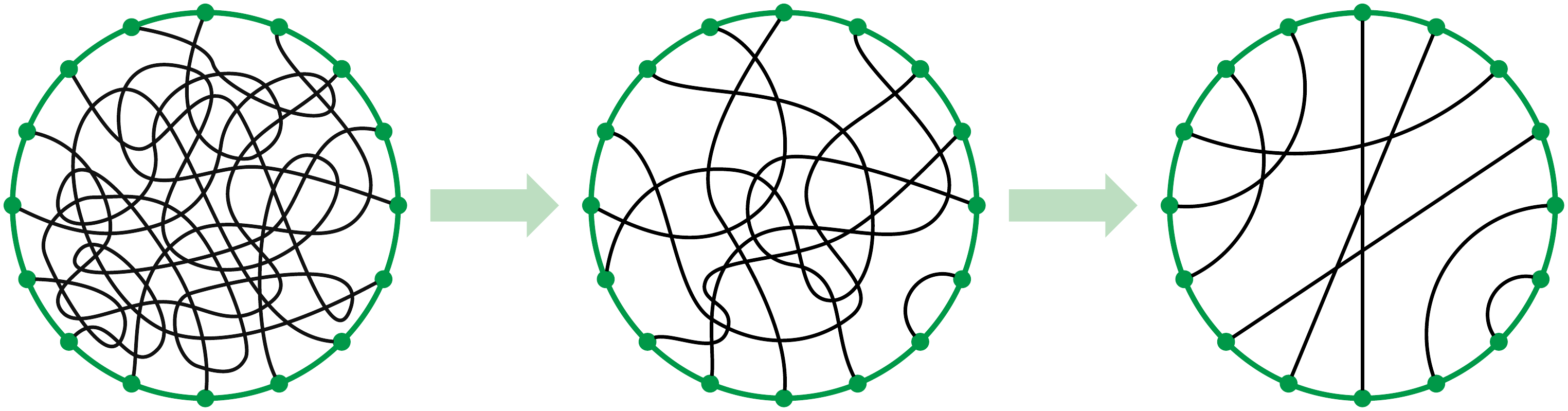}
\caption{Tightening a tangle in two phases: First simplifying the individual strands, then removing excess crossings between pairs of strands.}
\label{F:tighten-tangle}
\end{figure}

\begin{lemma}
\label{L:pretangle}
\EDIT{Every tangle $T$ with $n$ vertices and $s$ strands, such that every strand is simple, can be tightened using at most $3ns$ homotopy moves.}
\end{lemma}

\begin{proof}~
\EDIT{We prove the lemma by induction on $s$.  The base case when $s=1$ is trivial, so assume $s\ge 2$.}

Fix an arbitrary reference point on the boundary circle $\sigma$ that is not an endpoint of a strand. For each index $i$, let $\sigma_i$ be the arc of $\sigma$ between the endpoints of $\gamma_i$ that does not contain the reference point. A strand $\gamma_i$ is \emph{extremal} if the corresponding arc $\sigma_i$ does not contain any other arc $\sigma_j$.

Choose an arbitrary extremal strand $\gamma_i$. Let $m_i$ denote the number of tangle vertices in the interior of the disk bounded by $\gamma_i$ and the boundary arc $\sigma_i$; \EDIT{call this disk $\Sigma_i$}.  Let $s_i$ denote the number of intersections between $\gamma_i$ and other strands.  Finally, let $\gamma'_i$ be a path inside the disk $\Sigma$ defining tangle~$T$, with the same endpoints as $\gamma_i$, that intersects each other strand in $T$ at most once, such that the disk bounded by $\sigma_i$ and $\gamma'_i$ has no tangle vertices inside its interior. 
\EDIT{(See Figure~\ref{F:move-strand-tangle} for an illustration; the red strand in the left tangle is $\gamma_i$, the red strand in the middle tangle is $\gamma'_i$, and the shaded disk is $\Sigma_i$.)}.

We can deform $\gamma_i$ into $\gamma'_i$ using essentially the algorithm from Lemma \ref{L:contract}; \EDIT{the disk $\Sigma_i$ is contracted along with $\gamma_i$ in the process}. If \EDIT{$\Sigma_i$} contains an empty bigon \EDIT{with one side in $\gamma_i$}, remove it with a $\arc20$ move \EDIT{(which decreases~$s_i$ by $1$)}. If \EDIT{$\Sigma_i$} has an interior vertex with a neighbor on $\gamma_i$, remove it using at most two homotopy moves (which increases~$s_i$ \EDIT{by $1$ and decreases $m_i$ by $1$}). Altogether, this deformation requires at most $3m_i + s_i \le 3n$ homotopy moves.

\begin{figure}[ht]
\centering\includegraphics[width=0.9\linewidth]{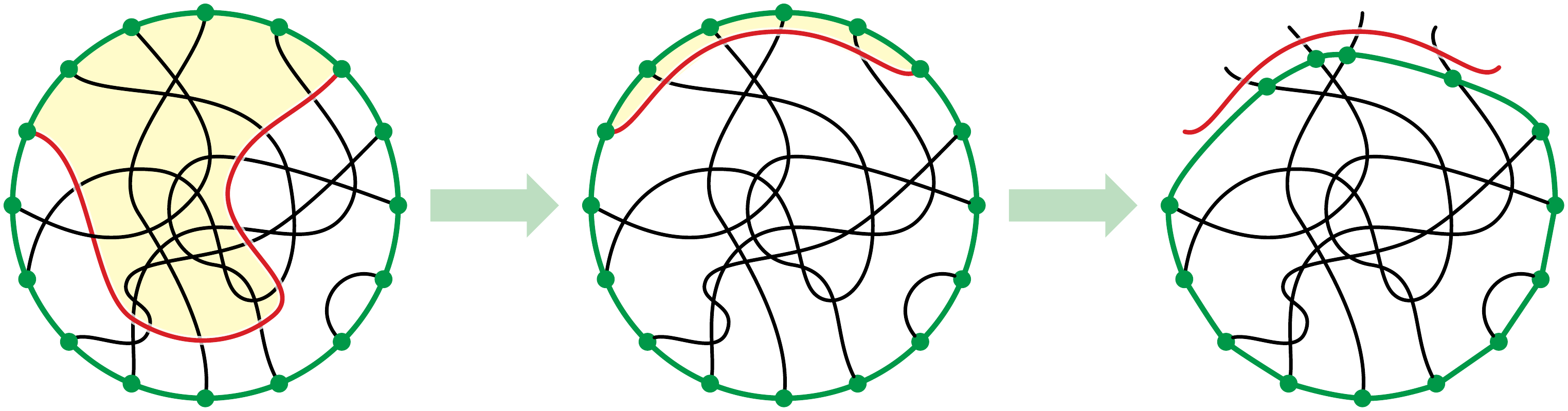}
\caption{Moving one strand out of the way and shrinking the tangle boundary.}
\label{F:move-strand-tangle}
\end{figure}

After deforming $\gamma_i$ to $\gamma_i'$, we \EDIT{redefine the tangle by ``shrinking'' its boundary curve} slightly to exclude~$\gamma'_i$, without creating or removing any \EDIT{vertices in the tangle or endpoints on the boundary} \EDIT{(see the right of Figure~\ref{F:move-strand-tangle})}.  We emphasize that shrinking the boundary does not modify the strands and therefore does not require any homotopy moves. The resulting smaller tangle has exactly $s-1$ strands, each of which is simple. Thus, the induction hypothesis implies that we can recursively tighten this smaller tangle using at most $3n(s-1)$ homotopy moves.  
\end{proof}

\begin{corollary}
\label{C:tangle}
Every tangle $T$ with $n$ vertices and $s$ strands can be tightened using at most $3D(T) + 3ns$ homotopy moves.
\end{corollary}

\begin{proof}
\EDIT{As long as $T$ contains at least one non-simple strand, we identify a simple loop $\alpha$ in that strand and remove it as described in the proof of Lemma \ref{L:contract}.  Suppose there are $m$ vertices and $t$ maximal subpaths in the interior of $\alpha$, and let $M$ be the number of homotopy moves required to remove $\alpha$.  The algorithm in the proof of Lemma \ref{L:contract} implies that $M\le 3m+t+1$, and Euler’s formula implies that $\alpha$ contains $m+t+1 \ge M/3$ faces.  Removing $\alpha$ decreases the depth of each of these faces by at least $1$ and therefore decreases the potential of the tangle by at least $M/3$.}

\EDIT{Let $T'$ be the remaining tangle after all such loops are removed.  Our potential analysis for a single loop implies inductively that transforming~$T$ into~$T'$ requires at most $3D(T) - 3D(T') \le 3D(T)$ homotopy moves.
Because $T'$ still has $s$ strands and at most $n$ vertices, Lemma \ref{L:pretangle} implies that we can tighten $T'$ with at most $3ns$ additional homotopy moves.}
\end{proof}

\subsection{Main Algorithm}

We call a simple closed curve $\sigma$ \EMPH{useful} for $\gamma$ if $\sigma$ intersects $\gamma$ transversely away from its vertices, and the interior tangle $T$ of $\sigma$ has at least $s^2$ vertices, where $s \coloneqq \abs{\sigma\cap\gamma}/2$ is the number of strands.
Our main algorithm repeatedly finds a useful closed curve whose interior tangle has $O(\sqrt{n})$ depth, and tightens its interior tangle;  if there are no useful closed curves, then we fall back to the loop-contraction algorithm of Lemma \ref{L:contract}.

%
%
%
%

\begin{lemma}
\label{L:useful}
Let $\gamma$ be an arbitrary non-simple closed curve in the plane with $n$ vertices.  Either there is a useful simple closed curve for $\gamma$ whose interior tangle has depth $O(\sqrt{n})$, or the depth of $\gamma$ is $O(\sqrt{n})$.
\end{lemma}

\begin{proof}
To simplify notation, let $d \coloneqq \Depth(\gamma)$. For each integer $j$ between $1$ and $d$, let $R_j$ be the set of points $p$ with $\Depth(p, \gamma) \ge d+1-j$, 
and let~$\tilde{R}_j$ denote a small open neighborhood of the closure of $R_j \cup \tilde{R}_{j-1}$, where $\tilde{R}_0$ is the empty set.  Each region~$\tilde{R}_j$ is the disjoint union of closed disks, whose boundary cycles intersect $\gamma$ transversely away from its vertices, if at all. In particular, $\tilde{R}_d$ is a disk containing the entire curve $\gamma$. 

Fix a point $z$ such that $\Depth(z,\gamma) = d$. For each integer $j$, let $\Sigma_j$ be the unique component of $\tilde{R}_j$ that contains $z$, and let $\sigma_j$ be the boundary of $\Sigma_j$. Then~$\sigma_1, \sigma_2, \dots, \sigma_d$ are disjoint, nested, simple closed curves; see Figure \ref{F:curve-levels}. Let $n_j$ be the number of vertices and let $s_j \coloneqq \abs{\gamma \cap \sigma_j}/2$ be the number of strands of the interior tangle of $\sigma_j$. For notational convenience, we define $\Sigma_0 \coloneqq \varnothing$ and thus $n_0 = s_0 = 0$.  We ignore the outermost curve $\sigma_d$, because it contains the entire curve $\gamma$.  The next outermost curve $\sigma_{d-1}$ contains every vertex of $\gamma$, so $n_{d-1} = n$.

\begin{figure}[ht]
\centering
\includegraphics[scale=0.5]{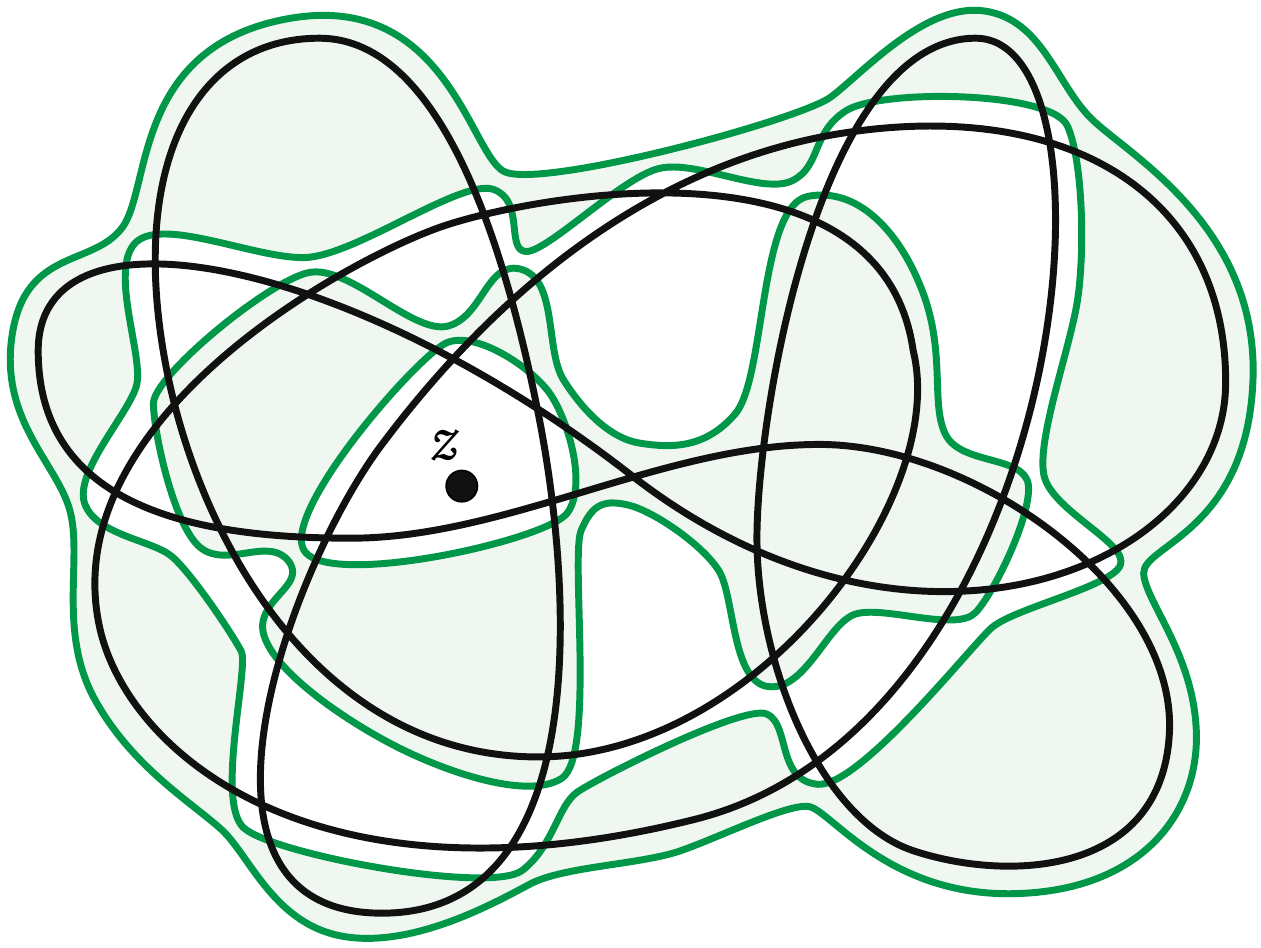}
\caption{Nested depth cycles around a point of maximum depth.}
\label{F:curve-levels}
\end{figure}

By construction, for each $j$, the interior tangle of $\sigma_j$ has depth $j+1$.  Thus, to prove the lemma, it suffices to show that if none of the curves $\sigma_1, \sigma_2, \dots, \sigma_{d-1}$ is useful, then $d = O(\sqrt{n})$.

Fix an index $j$. Each edge of $\gamma$ crosses $\sigma_j$ at most twice. Any edge of $\gamma$ that crosses~$\sigma_j$ has at least one endpoint in the annulus $\Sigma_j \setminus \Sigma_{j-1}$, and any edge that crosses $\sigma_j$ twice has both endpoints in~$\Sigma_j \setminus \Sigma_{j-1}$. Conversely, each vertex in $\Sigma_j$ is incident to at most two edges that cross $\sigma_j$ and no edges that cross $\sigma_{j+1}$. It follows that $\abs{\sigma_j\cap\gamma} \le 2(n_j - n_{j-1})$, and therefore $n_j \geq n_{j-1} + s_j$. 
Thus, by induction, we have
\[
	n_j \ge \sum_{i=1}^j s_i
\]
for every index $j$. 

Now suppose no curve $\sigma_j$ with $1\le j <d$ is useful.  Then we must have $s_j^2 > n_j$ and therefore 
\[
	s_j^2 > \sum_{i=1}^j s_i
\]
for all $1\le j < d$.  Trivially, $s_1\ge 1$, because $\gamma$ is non-simple. 
A straightforward induction argument implies that $s_j \ge (j+1)/2$
and therefore
\[
	n
	~=~ n_{d-1}
	~\ge~ \sum_{i = 1}^{d-1} \frac{i+1}{2}
	~\ge~ \frac{1}{2} \binom{d+1}{2}
	~>~ \frac{d^2}{4}.
\]
We conclude that $d \le 2\sqrt{n}$, which completes the proof.
\end{proof}

\begin{theorem}
\label{Th:upper}
Every closed curve in the plane with $n$ vertices can be simplified in $O(n^{3/2})$ homotopy moves.
\end{theorem}

\begin{proof}
Let $\gamma$ be an arbitrary closed curve in the plane with $n$ vertices.  If $\gamma$ has depth $O(\sqrt{n})$, Lemma~\ref{L:contract} and the trivial upper bound $D(\gamma) \le {(n+1)}\cdot \Depth(\gamma)$ imply that we can simplify $\gamma$ in $O(n^{3/2})$ homotopy moves. For purposes of analysis, we charge $O(\sqrt{n})$ of these moves to each vertex of $\gamma$. 
%

Otherwise, let $\sigma$ be an arbitrary useful closed curve chosen according to Lemma \ref{L:useful}. Suppose the interior tangle of $\sigma$ has $m$ vertices, $s$ strands, and depth~$d$.  Lemma \ref{L:useful} implies that $d = O(\sqrt{n})$, and the definition of useful implies that $s \le \sqrt{m}$, which is $O(\sqrt{n})$. Thus, by \EDIT{Corollary \ref{C:tangle}}, we can tighten the interior tangle of $\sigma$ in $O(m d + m s) = O(m \sqrt{n})$ moves. This simplification removes at least $m - s^2/2 \ge m/2$ vertices from $\gamma$, as the resulting tight tangle has at most $s^2/2$ vertices.  Again, for purposes of analysis, we charge $O(\sqrt{n})$ moves to each deleted vertex. We then recursively simplify the resulting closed curve.

In either case, each vertex of $\gamma$ is charged $O(\sqrt{n})$ moves as it is deleted. Thus, simplification requires at most $O(n^{3/2})$ homotopy moves in total.
\end{proof}

\subsection{Efficient Implementation}

Here we describe how to implement our curve-simplification algorithm to run in $O(n^{3/2})$ time; in fact, our implementation spends only constant amortized time per homotopy move.  We assume that the input curve is given in a data structure that allows fast exploration and modification of plane graphs, such as a quad-edge data structure \cite{gs-pmgsc-85} or a doubly-connected edge list \cite{bcko-cgaa-08}.  If the curve is presented as a polygon with $m$ edges, an appropriate graph representation can be constructed in $O(m\log m + n)$ time using classical geometric algorithms \cite{cs-arscg-89,m-fppa-90,ce-oails-92}; more recent algorithms can be used  for piecewise-algebraic curves \cite{ek-ee2aa-08}.

\begin{theorem}
\label{Th:upper-algo}
Given a simple closed curve $\gamma$ in the plane with $n$ vertices, we can compute a sequence of $M = O(n^{3/2})$ homotopy moves that simplifies $\gamma$ in $O(M)$ time.
\end{theorem}

\begin{proof}
We begin by labeling each face of \EDIT{$\gamma$} with its depth, using a breadth-first search of the dual graph in $O(n)$ time.  Then we construct the depth contours of \EDIT{$\gamma$}—the boundaries of the regions~$\tilde{R}_j$ from the proof of Lemma \ref{L:useful}—and organize them into a \emph{contour tree} in $O(n)$ time by brute force.  Another $O(n)$-time breadth-first traversal computes the number of strands and the number of interior vertices of every contour's interior tangle; in particular, we identify which depth contours are useful.  To complete the preprocessing phase, we place all the leafmost useful contours into a queue.  We can charge the overall $O(n)$ preprocessing time to the $\Omega(n)$ homotopy moves needed to simplify the curve.

As long as the queue of leafmost useful contours is non-empty, we extract one contour $\sigma$ from this queue and simplify its interior tangle $T$ as follows.  Suppose $T$ has $m$ interior vertices.

Following the proof of Theorem \ref{Th:upper}, we first simplify every loop in each strands of $T$.  We identify loops by traversing the strand from one endpoint to the other, marking the vertices as we go; the first time we visit a vertex that has already been marked, we have found a loop $\alpha$.  We can perform each of the homotopy moves required to shrink $\alpha$ in $O(1)$ time, because each such move modifies only a constant-radius boundary of a vertex on $\alpha$.  After the loop is shrunk, we continue walking along the strand starting at the most recently marked \EDIT{vertex}.

The second phase of the tangle-simplification algorithm proceeds similarly.  We walk around the boundary of $T$, marking vertices as we go.  As soon as we see the second endpoint of any strand $\gamma_i$, we pause the walk to straighten $\gamma_i$.  As before, we can execute each homotopy move used to move $\gamma_i$ to $\gamma'_i$ in $O(1)$ time. We then move the boundary of the tangle over the vertices of $\gamma'_i$, and remove the endpoints of $\gamma'_i$ from the boundary curve, in $O(1)$ time per vertex.

The only portions of the running time that we have not already charged to homotopy moves are the time spent marking the vertices on each strand and the time to update the tangle boundary after moving a strand aside.  Altogether, the uncharged time is $O(m)$, which is less than the number of moves used to tighten $T$, because the contour $\sigma$ is useful.  Thus, tightening the interior tangle of a useful contour requires $O(1)$ amortized time per homotopy move.

Once the tangle is tight, we must update the queue of useful contours.  The original contour $\sigma$ is still a depth contour in the modified curve, and tightening $T$ only changes the depths of faces that intersect~$T$.  Thus, we could update the contour tree in $O(m)$ time, which we could charge to the moves used to tighten $T$; but in fact, this update is unnecessary, because no contour in the interior of $\sigma$ is useful.  We then walk up the contour tree from $\sigma$, updating the number of interior vertices until we find a useful ancestor contour.  The total time spent traversing the contour tree for new useful contours is $O(n)$; we can charge this time to the $\Omega(n)$ moves needed to simplify the curve.
\end{proof}

\subsection{Multicurves}

Finally, we describe how to extend our $O(n^{3/2})$ upper bound to multicurves.  Just as in Section~\ref{SS:multi-lower}, we distinguish between two variants, depending on whether the target of the simplification is an \emph{arbitrary} set of disjoint cycles or a \emph{particular} set of disjoint cycles.  In both cases, our upper bounds match the lower bounds proved in Section~\ref{SS:multi-lower}.

First we extend our loop-contraction algorithm from Lemma \ref{L:contract} to the multicurve setting.  
\EDIT{Recall that a \emph{component} of a multicurve $\gamma$ is any multicurve whose image is a component of the image of $\gamma$, and the individual closed curves that comprise $\gamma$ are its \emph{constituent curves}.}
The main difficulty is that \EDIT{one} component of the multicurve might lie inside a face of another component, making progress on the larger component impossible.  To handle this potential obstacle, we simplify the \emph{innermost} components of the \EDIT{multicurve} first, and we move isolated simple closed curves toward the outer face as quickly as possible.  Figure \ref{F:multi-shrink} sketches the basic steps of our algorithm when the input multicurve is connected.

\begin{figure}[ht]
\centering
\includegraphics[width=0.95\linewidth]{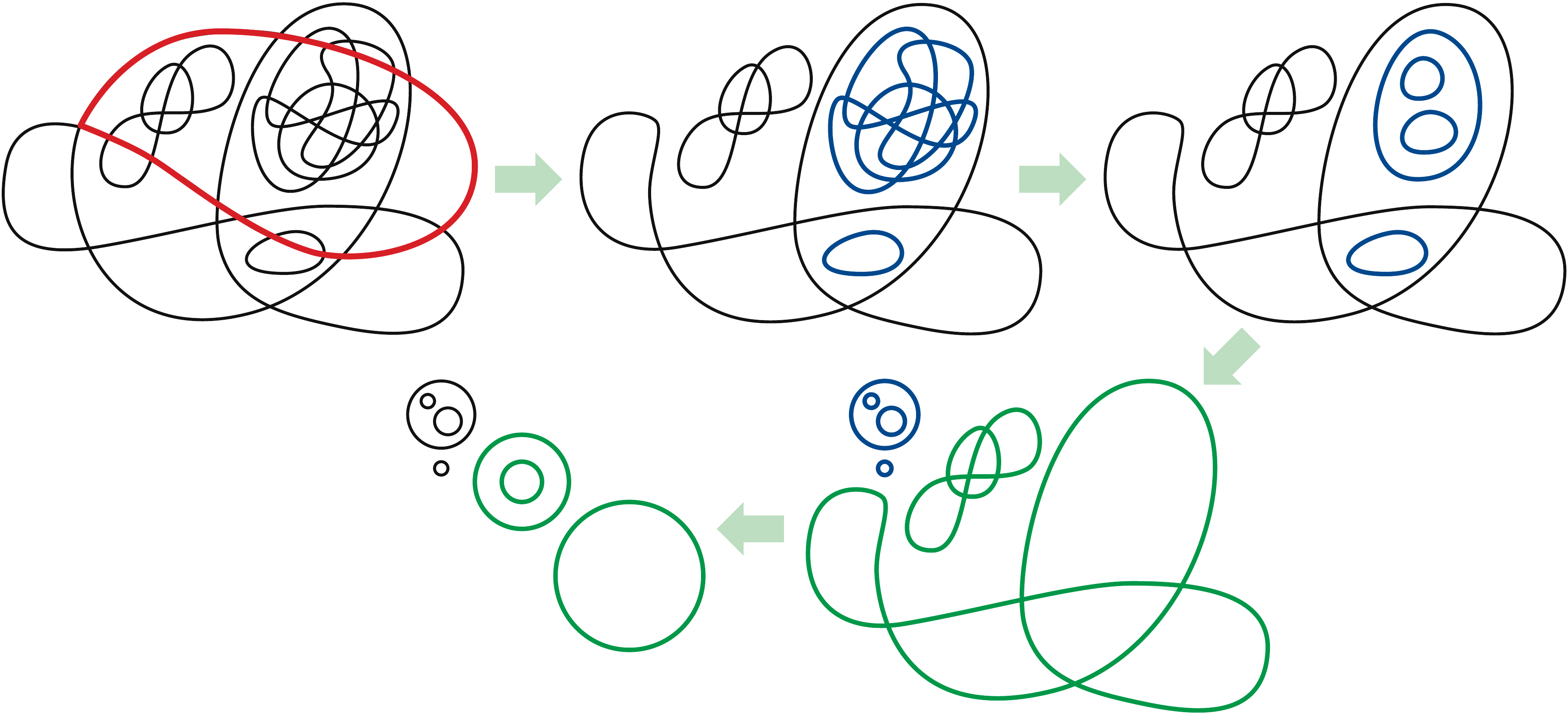}
\caption{Simplifying a connected multicurve: shrink an arbitrary simple loop or cycle, recursively simplify any inner components, translate inner circle clusters to the outer face, and recursively simplify the remaining non-simple components.}
\label{F:multi-shrink}
\end{figure}
\unskip
\begin{lemma}
\label{L:multi-contract}
Every $n$-vertex $k$-curve $\gamma$ in the plane can be transformed into $k$ disjoint simple closed curves using at most $3D(\gamma) + 4nk$ homotopy moves.
\end{lemma}

\begin{proof}
Let $\gamma$ be an arbitrary $k$-curve with $n$ vertices.
If $\gamma$ is connected, we either contract and delete a loop, exactly as in Lemma \ref{L:contract}, or we contract a simple constituent curve to an isolated circle, using essentially the same algorithm.  In either case, the number of moves performed is at most $3D(\gamma) - 3D(\gamma')$, where $\gamma'$ is the multicurve after the contraction.  The lemma now follows immediately by induction.

We call a component of $\gamma$ an \EMPH{outer component} if it is incident to the unbounded outer face of $\gamma$, and an \EMPH{inner component} otherwise.  If $\gamma$ has more than one outer component, we partition $\gamma$ into subcurves, each consisting of one outer component $\gamma\!_o$ and all inner components located inside faces of $\gamma\!_o$, and we recursively simplify each subcurve independently; the lemma follows by induction.  If any outer component is simple, we ignore that component and simplify the rest of $\gamma$ recursively; again, the lemma follows by induction.

Thus, we can assume without loss of generality that our multicurve $\gamma$ is disconnected but has only one outer component $\gamma\!_o$, which is non-simple.  For each face $f$ of $\gamma\!_o$, let $\gamma\!_f$ denote the union of all components inside $f$.  Let $n_f$ and $k_f$ respectively denote the number of vertices and constituent \EDIT{curves} of~$\gamma\!_f$.  Similarly, let~$n_o$ and $k_o$ respectively denote the number of vertices and constituent \EDIT{curves} of the outer component $\gamma\!_o$.

We first recursively simplify each subcurve $\gamma\!_f$; let $\kappa_f$ denote the resulting \emph{cluster} of $k_f$ simple closed curves.  By the induction hypothesis, this simplification requires at most $3D(\gamma\!_f) + 4 n_f k_f$ homotopy moves.  We \emph{translate} each cluster $\kappa_f$ to the outer face of $\gamma\!_o$ by shrinking $\kappa_f$ to a small $\e$-ball and then moving the entire cluster along a shortest path in the dual graph of $\gamma\!_o$.  This translation requires at most $4n_o k_f $ homotopy moves; each circle in $\kappa_f$ uses one $\arc{2}{0}$ move and one $\arc{0}{2}$ move to cross any edge of $\gamma\!_o$, and in the worst case, the cluster might cross all $2n_0$ edges of $\gamma\!_o$.  After all circle clusters are in the outer face, we recursively simplify $\gamma\!_o$ using at most $3 D(\gamma\!_o) + 4 n_o k_o$ homotopy moves.  

The total number of homotopy moves used in this case is
\[
	\sum_f 3D(\gamma\!_f) + 3 D(\gamma\!_o)
	~+~
	\sum_f 4 n_f k_f + \sum_f 4 n_o k_f + 4 n_o k_o.
\]
Each face of $\gamma\!_o$ has the same depth as the corresponding face of $\gamma$, and for each face~$f$ of $\gamma\!_o$, each face of the subcurve $\gamma\!_f$ has lesser depth than the corresponding face of $\gamma$.  It follows that 
\[
	\sum_f D(\gamma\!_f) + D(\gamma\!_o) \le D(\gamma).
\]
Similarly, $\sum_f n_f + n_o = n$ and $\sum_f k_f + k_o = k$.  The lemma now follows immediately.
\end{proof}

To reduce the leading term to $O(n^{3/2})$, we extend the definition of a tangle to the intersection of a multicurve $\gamma$ with a closed disk whose boundary intersects the multicurve transversely away from its vertices, or not at all.  Such a tangle can be decomposed into boundary-to-boundary paths, called \emph{open} strands, and closed curves that do not touch the tangle boundary, called \emph{closed} strands.  \EDIT{Each closed strand is a constituent curve of $\gamma$.}  A tangle is \emph{tight} if every strand is simple, every pair of open strands intersects at most once, and otherwise all strands are disjoint.

\begin{theorem}
\label{Th:multi-upper}
Every $k$-curve in the plane with $n$ vertices can be transformed into a set of $k$ disjoint simple closed curves using $O(n^{3/2} + nk)$ homotopy moves.
\end{theorem}

\begin{proof}
Let $\gamma$ be an arbitrary $k$-curve with $n$ vertices.  Following the proof of Lemma \ref{L:multi-contract}, we can assume without loss of generality that $\gamma$ has a single outer component $\gamma\!_o$, which is non-simple.

When $\gamma$ is disconnected, we follow the strategy in the previous proof.  Let $\gamma\!_f$ denote the union of all components inside any face $f$ of $\gamma\!_o$.  For each face $f$, we recursively simplify $\gamma\!_f$ and translate the resulting cluster of disjoint circles to the outer face; when all faces are empty, we recursively simplify $\gamma\!_o$.  The theorem now follows by induction.

When $\gamma$ is non-simple and connected, we follow the useful closed curve strategy from Theorem~\ref{Th:upper}. We define a closed curve $\sigma$ to be useful for $\gamma$ if the interior tangle of $\sigma$ has its number of vertices at least the square of the number of \emph{open} strands; then the proof of Lemma~\ref{L:useful} applies to connected multicurves with no modifications.  So let $T$ be a tangle with $m$ vertices, $s \le \sqrt{m}$ open strands, $\ell$ closed strands, and depth $d = O(\sqrt{n})$.  We straighten~$T$ in two phases, almost exactly as in \EDIT{Section~\ref{SS:tangles}}, contracting loops and simple closed strands in the first phase, and straightening open strands in the second phase.

In the first phase, contracting one loop or \EDIT{simple} closed \EDIT{strand} uses at most $3D(T) - 3D(T')$ homotopy moves, where $T'$ is the tangle after contraction.  After each contraction, if $T'$ is disconnected—in particular, if we just contracted a \EDIT{simple} closed \EDIT{strand}—we simplify and extract any isolated components as follows.  Let $T'_o$ denote the component of $T'$ \EDIT{that} includes the boundary cycle, and for each face $f$ of $T'_o$, let $\gamma_f$ denote the union of all components of $T'$ inside $f$.  We simplify each multicurve $\gamma_f$ using the algorithm from Lemma~\ref{L:multi-contract}—\emph{not recursively!}—and then translate the resulting cluster of disjoint circles \emph{to the outer face of~$\gamma$}.  \EDIT{See Figure \ref{F:tighten-open}.}  Altogether, simplifying and translating these subcurves requires at most
\(
	3D(T') - 3D(T'') + 4n \sum_f k_f
\)
homotopy moves, where $T''$ is the resulting tangle.

\begin{figure}[ht]
\centering
\includegraphics[width=\linewidth]{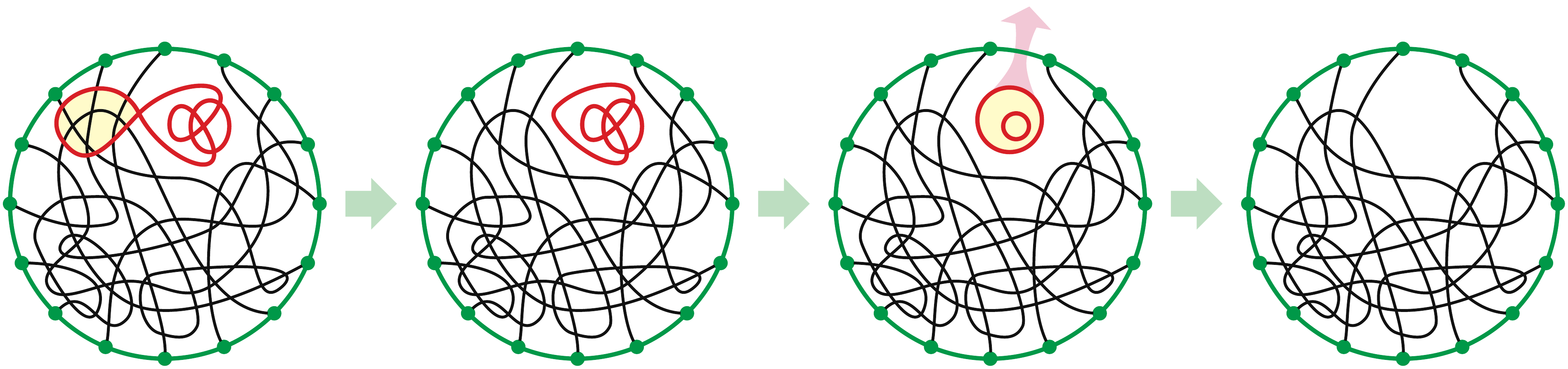}
\caption{Whenever shrinking a loop or simple closed strand disconnects the tangle, simplify each isolated component and translate the resulting cluster of circles to the outer face of the entire multicurve.}
\label{F:tighten-open}
\end{figure}

The total number of moves performed in the first phase is at most $3D(T) + 4m\ell = O(m\sqrt{n} + n\ell)$.  The first phase ends when the tangle consists entirely of simple open strands.  Thus, the second phase straightens the remaining open strands exactly as in the proof of \EDIT{Lemma \ref{L:pretangle}}; the total number of moves in this phase is $O(ms) = O(m\sqrt{n})$.  We charge $O(\sqrt{n})$ time to each deleted vertex and $O(n)$ time to each constituent curve that was simplified and translated outward.  We then recursively simplify the remaining multicurve, ignoring any outer circle clusters.

Altogether, each vertex of $\gamma$ is charged $O(\sqrt{n})$ time as it is deleted, and each constituent curve of $\gamma$ is charged $O(n)$ time as it is translated outward.
\end{proof}

With $O(k^2)$ additional homotopy moves, we can transform the resulting set of $k$ disjoint circles into $k$ nested or unnested circles.

\begin{theorem}
\label{Th:multi-upper2}
Any $k$-curve with $n$ vertices in the plane can be transformed into $k$ nested (or unnested) simple closed curves using $O(n^{3/2} + nk + k^2)$ homotopy moves.
\end{theorem}

\begin{corollary}
\label{C:multi-upper3}
Any $k$-curve with at most $n$ vertices in the plane can be transformed into any other $k$-curve with at most $n$ vertices using $O(n^{3/2} + nk + k^2)$ homotopy moves.
\end{corollary}

Theorems \ref{Th:multi-lower} and \ref{Th:multi-lower2} and Corollary \ref{C:multi-lower3} imply that these upper bounds are tight in the worst case for all possible values of $n$ and $k$.  As in the lower bounds, the $O(k^2)$ terms are redundant for connected multicurves.

More careful analysis implies that any $k$-curve with $n$ vertices and depth $d$ can be simplified in $O(n \min\set{d, n^{1/2}} + k\min\set{d, n})$ homotopy moves, or transformed into $k$ unnested circles using $O(n \min\set{d, n^{1/2}} + k \min\set{d, n} + k \min\set{d, k})$ homotopy moves.  Moreover, these upper bounds are tight for all possible values of $n$, $k$, and $d$.  We leave the details of this extension as an exercise for the reader.

\section{Higher-Genus Surfaces}
\label{S:genus}

Finally, we consider the natural generalization of our problem to closed curves on orientable surfaces of higher genus. Because these surfaces have non-trivial topology, not every closed curve is homotopic to a single point or even to a simple curve. A closed curve is \EMPH{contractible} if it is homotopic to a single point. We call a closed curve \EMPH{tight} if it has the minimum number of self-intersections in its homotopy class.

\subsection{Lower Bounds}

Although defect was originally defined as an invariant of \emph{planar} curves, Polyak's formula $\Defect(\gamma) = -2\sum_{x\between y} \sgn(x)\sgn(y)$ extends naturally to closed curves on any orientable surface; homotopy moves change the resulting invariant exactly as described in Figure \ref{F:defect-change}. Thus, Lemma~\ref{L:defect} immediately generalizes to any orientable surface as follows.

\begin{lemma}
\label{L:defect-surface}
Let $\gamma$ and $\gamma'$ be arbitrary closed curves that are homotopic on an arbitrary orientable surface. Transforming $\gamma$ into $\gamma'$ requires at least $\abs{\Defect(\gamma) - \Defect(\gamma')}/2$ homotopy moves.
\end{lemma}

The following construction implies a quadratic lower bound for tightening noncontractible curves on orientable surfaces with any positive genus.

\begin{lemma}
For any positive integer $n$, there is a closed curve on the torus with $n$ vertices and defect $\Omega(n^2)$ that is homotopic to a simple closed curve but not contractible.
\end{lemma}

\begin{proof}
Without loss of generality, suppose $n$ is a multiple of $8$. The curve $\gamma$ is illustrated on the left in Figure \ref{F:bad-torus}. The torus is represented by a rectangle with opposite edges identified. We label three points $a,b,c$ on the vertical edge of the rectangle and decompose the curve into a red path from $a$ to $b$, a green path from $b$ to $c$, and a blue path from~$c$ to $a$. The red and blue paths each wind vertically around the torus, first $n/8$ times in one direction, and then $n/8$ times in the opposite direction.

\begin{figure}[ht]
\centering\includegraphics[scale=0.5]{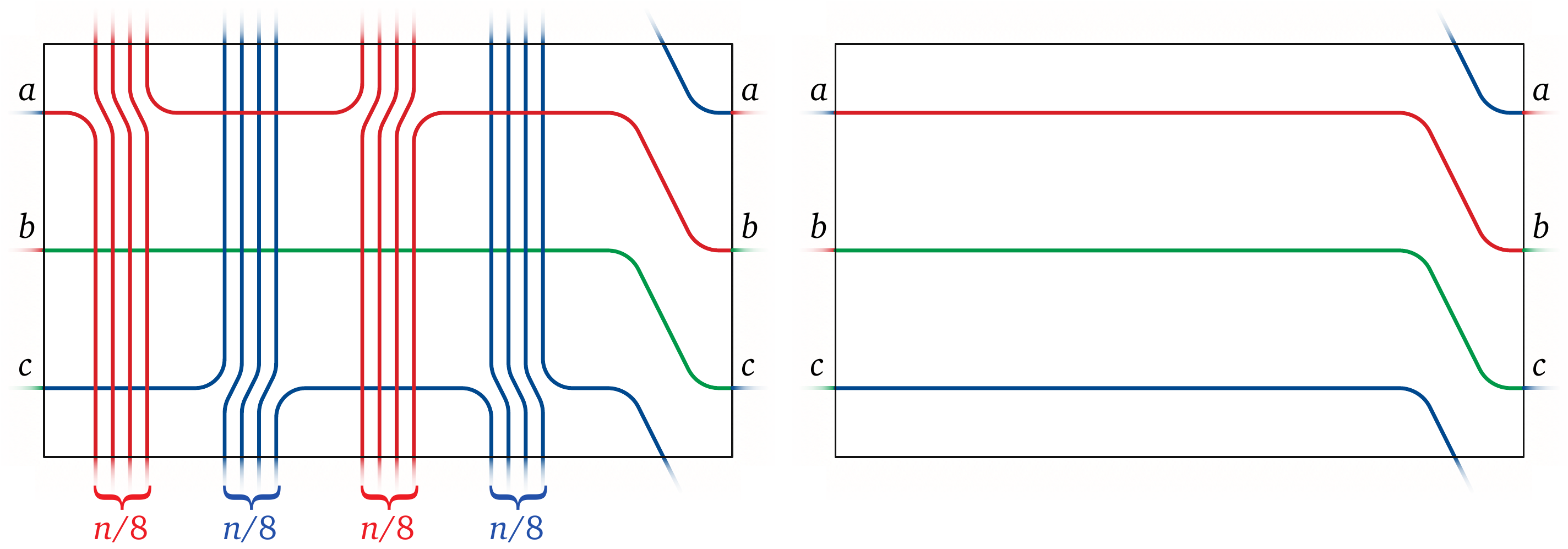}
\caption{A curve $\gamma$ on the torus with defect $\Omega(n^2)$ and a simple curve homotopic to $\gamma$.}
\label{F:bad-torus}
\end{figure}

As in previous proofs, we compute the defect of $\gamma$ by describing a sequence of homotopy moves that \EDIT{simplifies} the curve, while carefully tracking the changes in the defect that these moves incur. We can unwind one turn of the red path by performing one $\arc20$ move, followed by $n/8$ $\arc33$ moves, followed by one $\arc20$ move, as illustrated in Figure \ref{F:bad-torus-unwind}. Repeating this sequence of homotopy moves $n/8$ times removes all intersections between the red and green paths, after which a sequence of $n/4$ $\arc20$ moves straightens the blue path, yielding the simple curve shown on the right in Figure \ref{F:bad-torus}. Altogether, we perform $n^2/64 + n/2$ homotopy moves, where each $\arc33$ move increases the defect of the curve by $2$ and each $\arc20$ move decreases the defect of the curve by $2$. We conclude that $\Defect(\gamma) = -n^2/32 + n$.
\end{proof}

\begin{figure}[ht]
\centering\includegraphics[width=0.9\linewidth]{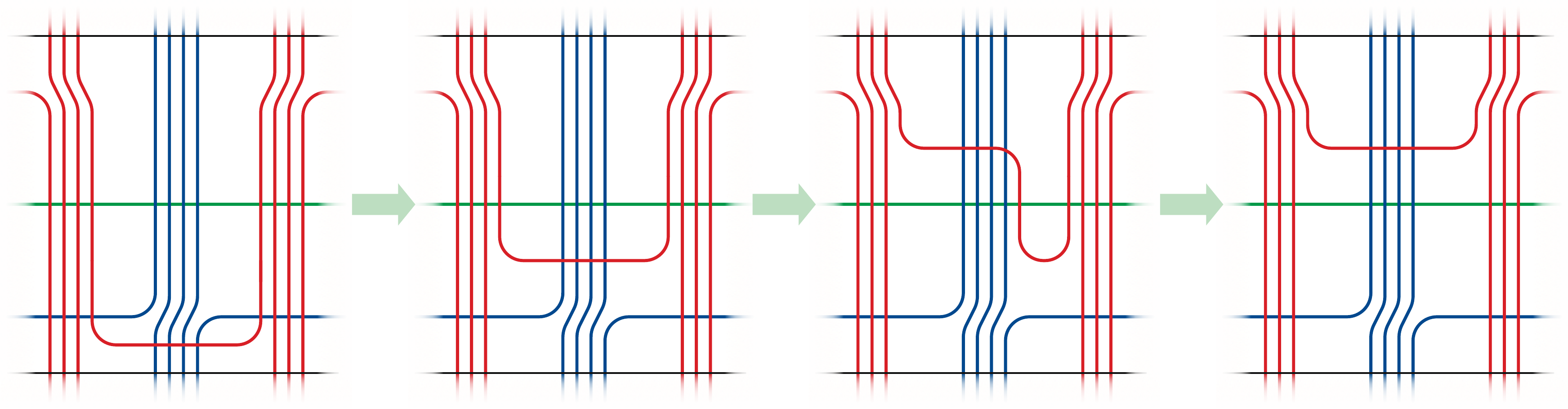}
\caption{Unwinding one turn of the red path.}
\label{F:bad-torus-unwind}
\end{figure}

\begin{theorem}
\label{Th:lower-torus}
Tightening a closed curve with $n$ crossings on a torus requires $\Omega(n^2)$ homotopy moves in the worst case, even if the curve is homotopic to a simple curve.
\end{theorem}


\subsection{Upper Bounds}

Hass and Scott proved that any non-simple closed curve on any orientable surface that is homotopic to a simple closed curve contains either a simple (in fact empty) contractible loop or a simple contractible bigon \cite[Theorem 1]{hs-ics-85}. It follows immediately that any such curve can be simplified in $O(n^2)$ moves using Steinitz's algorithm; Theorem \ref{Th:lower-torus} implies that the upper bound is tight for non-contractible curves.

For the most general setting, where the given curve is not necessarily homotopic to a simple closed curve, we are not even aware of a \emph{polynomial} upper bound!  Steinitz's algorithm does not work here; there are curves with excess self-intersections but no simple contractible loops or bigons \cite{hs-ics-85}. Hass and Scott \cite{hs-scs-94} and De Graff and Schrijver \cite{gs-mcmcr-97} independently proved that any closed curve on any surface can be simplified using a \emph{finite} number of homotopy moves that never increase the number of self-intersections. Both proofs use discrete variants of curve-shortening flow; for sufficiently well-behaved curves and surfaces, results of Grayson \cite{g-sec-89} and Angenent \cite{a-pecs2-91} imply a similar result for differential curvature flow. Unfortunately, without further assumptions about the precise geometries of both the curve and the underlying surface, the number of homotopy moves cannot be bounded by any function of the number of crossings; even in the plane, there are closed curves with three crossings for which curve-shortening flow alternates between a $\arc33$ move and its inverse arbitrarily many times. Paterson \cite{p-cails-02} describes a combinatorial algorithm to compute a tightening sequence of homotopy moves without such reversals, but she offers no analysis of her algorithm.

We conjecture that any arbitrary curves (or even multicurves) on any surface can be simplified with at most $O(n^2)$ homotopy moves.

%
%
%
%
%
%

\paragraph{Acknowledgements.}
We would like to thank Nathan Dunfield, Joel Hass, Bojan Mohar, and Bob Tarjan for encouragement and helpful discussions, and Joe O'Rourke for asking about multiple curves.  \EDIT{We would also like to thank the anonymous referees for several suggestions that improved the paper.}

\bibliographystyle{newuser}
\bibliography{bib/jeffe,bib/topology,bib/compgeom}

\def\burl#1{$\langle$\url{#1}$\rangle$}
\begin{thebibliography}{10}

\bibitem{a-tc-94}
Francesca Aicardi.
\newblock Tree-like curves.
\newblock \emph{Singularities and Bifurcations}, 1--31, 1994. Advances in
  Soviet Mathematics~21, Amer. Math. Soc.

\bibitem{a-wtns-60}
Sheldon~B. {Akers, Jr.}
\newblock The use of wye-delta transformations in network simplification.
\newblock \emph{Oper. Res.} 8(3):311--323, 1960.

\bibitem{a-cas-26}
James~W. Alexander.
\newblock Combinatorial analysis situs.
\newblock \emph{Trans. Amer. Math. Soc.} 28(2):301--326, 1926.

\bibitem{ab-tkc-26}
James~W. Alexander and G.~B. Briggs.
\newblock On types of knotted curves.
\newblock \emph{Ann. Math.} 28(1/4):562--586, 1926--1927.

\bibitem{abil-ivd-16}
Sarah~R. Allen, Luis Barba, John Iacono, and Stefan Langerman.
\newblock {Incremental Voronoi diagrams}.
\newblock \emph{Proc. 32nd Int. Symp. Comput. Geom.}, 15:1--15:16, 2016.
  Leibniz International Proceedings in Informatics~51.
\newblock \burl{http://drops.dagstuhl.de/opus/volltexte/2016/5907}.
\newblock arXiv:\href{http://arxiv.org/abs/1603.08485}{1603.08485}.

\bibitem{a-pecs2-91}
Sigurd Angenent.
\newblock Parabolic equations for curves on surfaces: {Part} {II}.
  {Intersections}, blow-up and generalized solutions.
\newblock \emph{Ann. Math.} 133(1):171--215, 1991.

\bibitem{acgp-frpwg-00}
Dan Archdeacon, Charles~J. Colbourn, Isidoro Gitler, and J.~Scott Provan.
\newblock Four-terminal reducibility and projective-planar
  wye-delta-wye-reducible graphs.
\newblock \emph{J. Graph Theory} 33(2):83--93, 2000.

\bibitem{a-pctip-94}
Vladimir~I. Arnold.
\newblock Plane curves, their invariants, perestroikas and classifications.
\newblock \emph{Singularities and Bifurcations}, 33--91, 1994. Adv. Soviet
  Math.~21, Amer. Math. Soc.

\bibitem{a-tipcc-94}
Vladimir~I. Arnold.
\newblock \emph{Topological Invariants of Plane Curves and Caustics}.
\newblock University Lecture Series~5. Amer. Math. Soc., 1994.

\bibitem{bcko-cgaa-08}
Mark de~Berg, Otfried Cheong, Marc van Kreveld, and Mark Overmars.
\newblock \emph{Computational Geometry: Algorithms and Applications}, 3rd
  edition.
\newblock Springer-Verlag, 2008.

\bibitem{defect}
Hsien-Chih Chang and Jeff Erickson.
\newblock Electrical reduction, homotopy moves, and defect.
\newblock Preprint, October 2015.
\newblock arXiv:\href{http://arxiv.org/abs/1510.00571}{1510.00571}.

\bibitem{tangle}
Hsien-Chih Chang and Jeff Erickson.
\newblock Untangling planar curves.
\newblock \emph{Proc. 32nd Int. Symp. Comput. Geom.}, 29:1--29:15, 2016.
  Leibniz International Proceedings in Informatics~51.
\newblock \burl{http://drops.dagstuhl.de/opus/volltexte/2016/5921}.

\bibitem{ce-oails-92}
Bernard Chazelle and Herbert Edelsbrunner.
\newblock An optimal algorithm for intersecting line segments in the plane.
\newblock \emph{J. ACM} 39(1):1--54, 1992.

\bibitem{cdm-ivki-12}
Sergei Chmutov, Sergei Duzhin, and Jacob Mostovoy.
\newblock \emph{Introduction to {Vassiliev} knot invariants}.
\newblock Cambridge Univ. Press, 2012.
\newblock \burl{http://www.pdmi.ras.ru/~duzhin/papers/cdbook}.
\newblock arXiv:\href{http://arxiv.org/abs/1103.5628}{1103.5628}.

\bibitem{cs-arscg-89}
Kenneth~L. Clarkson and Peter~W. Shor.
\newblock Applications of random sampling in computational geometry, {II}.
\newblock \emph{Discrete Comput. Geom.} 4:387--421, 1989.

\bibitem{cpv-nastc-95}
Charles~J. Colbourn, J.~Scott Provan, and Dirk Vertigan.
\newblock A new approach to solving three combinatorial enumeration problems on
  planar graphs.
\newblock \emph{Discrete Appl. Math.} 60:119--129, 1995.

\bibitem{cgv-rep-96}
Yves {Colin de Verdière}, Isidoro Gitler, and Dirk Vertigan.
\newblock Réseaux électriques planaires {II}.
\newblock \emph{Comment. Math. Helvetici} 71:144--167, 1996.

\bibitem{c-eklst-70}
John~H. Conway.
\newblock An enumeration of knots and links, and some of their algebraic
  properties.
\newblock \emph{Computational Problems in Abstract Algebra}, 329--358, 1970.
  Pergamon Press.

\bibitem{cl-ubrm-14}
Alexander Coward and Marc Lackenby.
\newblock An upper bound on {Reidemeister} moves.
\newblock \emph{Amer. J. Math.} 136(4):1023--1066, 2014.
\newblock arXiv:\href{http://arxiv.org/abs/1104.1882}{1104.1882}.

\bibitem{ek-ee2aa-08}
Arno Eigenwillig and Michael Kerber.
\newblock Exact and efficient {2D}-arrangements of arbitrary algebraic curves.
\newblock \emph{Proc. 19th Ann. ACM-SIAM Symp. Discrete Algorithms}, 122--131,
  2008.

\bibitem{e-rpges-66}
G.~{V}. Epifanov.
\newblock Reduction of a plane graph to an edge by a star-triangle
  transformation.
\newblock \emph{Dokl. Akad. Nauk SSSR} 166:19--22, 1966.
\newblock In Russian. English translation in \textit{Soviet Math. Dokl.}
  7:13--17, 1966.

\bibitem{ehln-irkl-14}
Chaim Even-Zohar, Joel Hass, Nati Linial, and Tahl Nowik.
\newblock Invariants of random knots and links.
\newblock \emph{Discrete \& Computational Geometry} 56(2):274--314, 2016.
\newblock arXiv:\href{http://arxiv.org/abs/1411.3308}{1411.3308}.

\bibitem{f-erpns-85}
Thomas~A. Feo.
\newblock \emph{{I.} {A} {Lagrangian} Relaxation Method for Testing The
  Infeasibility of Certain {VLSI} Routing Problems. {II.} {Efficient} Reduction
  of Planar Networks For Solving Certain Combinatorial Problems}.
\newblock Ph.D. thesis, Univ. California Berkeley, 1985.
\newblock \burl{http://search.proquest.com/docview/303364161}.

\bibitem{fp-dtert-93}
Thomas~A. Feo and J.~Scott Provan.
\newblock Delta-wye transformations and the efficient reduction of two-terminal
  planar graphs.
\newblock \emph{Oper. Res.} 41(3):572--582, 1993.

\bibitem{f-frtcs-69}
George~K. Francis.
\newblock The folded ribbon theorem: {A} contribution to the study of immersed
  circles.
\newblock \emph{Trans. Amer. Math. Soc.} 141:271--303, 1969.

\bibitem{g-n1gs-00}
Carl~Friedrich Gauß.
\newblock Nachlass. {I.} {Zur} {Geometria} situs.
\newblock \emph{Werke}, vol.~8, 271--281, 1900. Teubner.
\newblock Originally written between 1823 and 1840.

\bibitem{g-dtaa-91}
Isidoro Gitler.
\newblock \emph{Delta-wye-delta Transformations: Algorithms and Applications}.
\newblock Ph.D. thesis, Department of Combinatorics and Optimization,
  University of Waterloo, 1991.

\bibitem{gs-mcmcr-97}
Maurits de~Graaf and Alexander Schrijver.
\newblock Making curves minimally crossing by {Reidemeister} moves.
\newblock \emph{J. Comb. Theory Ser. B} 70(1):134–156, 1997.

\bibitem{g-sec-89}
Matthew~A. Grayson.
\newblock Shortening embedded curves.
\newblock \emph{Ann. Math.} 129(1):71--111, 1989.

\bibitem{g-cp-67}
Branko Grünbaum.
\newblock \emph{Convex Polytopes}.
\newblock Monographs in Pure and Applied Mathematics XVI. John Wiley \& Sons,
  1967.

\bibitem{gs-pmgsc-85}
Leonidas~J. Guibas and Jorge Stolfi.
\newblock Primitives for the manipulation of general subdivisions and the
  computation of {Voronoi} diagrams.
\newblock \emph{ACM Trans. Graphics} 4(2):75--123, 1985.

\bibitem{hn-udrqn-10}
Joel Hass and Tal Nowik.
\newblock Unknot diagrams requiring a quadratic number of {Reidemeister} moves
  to untangle.
\newblock \emph{Discrete Comput. Geom.} 44(1):91--95, 2010.

\bibitem{hs-ics-85}
Joel Hass and Peter Scott.
\newblock Intersections of curves on surfaces.
\newblock \emph{Israel J. Math.} 51:90--120, 1985.

\bibitem{hs-scs-94}
Joel Hass and Peter Scott.
\newblock Shortening curves on surfaces.
\newblock \emph{Topology} 33(1):25--43, 1994.

\bibitem{hh-msrmd-11}
Chuichiro Hayashi and Miwa Hayashi.
\newblock Minimal sequences of {Reidemeister} moves on diagrams of torus knots.
\newblock \emph{Proc. Amer. Math. Soc.} 139:2605--2614, 2011.
\newblock arXiv:\href{http://arxiv.org/abs/1003.1349}{1003.1349}.

\bibitem{hhn-unnrm-12}
Chuichiro Hayashi, Miwa Hayashi, and Tahl Nowik.
\newblock Unknotting number and number of {Reidemeister} moves needed for
  unlinking.
\newblock \emph{Topology Appl.} 159:1467--1474, 2012.
\newblock arXiv:\href{http://arxiv.org/abs/1012.4131}{1012.4131}.

\bibitem{hhsy-musrm-12}
Chuichiro Hayashi, Miwa Hayashi, Minori Sawada, and Sayaka Yamada.
\newblock Minimal unknotting sequences of {Reidemeister} moves containing
  unmatched {RII} moves.
\newblock \emph{J. Knot Theory Ramif.} 21(10):1250099 (13 pages), 2012.
\newblock arXiv:\href{http://arxiv.org/abs/1011.3963}{1011.3963}.

\bibitem{h-udtse-35}
Heinz Hopf.
\newblock Über die {Drehung} der {Tangenten} und {Sehnen} ebener {Kurven}.
\newblock \emph{Compositio Math.} 2:50--62, 1935.

\bibitem{it-whkp-13}
Noburo Ito and Yusuke Takimura.
\newblock (1,2) and weak (1,3) homotopies on knot projections.
\newblock \emph{J. Knot Theory Ramif.} 22(14):1350085 (14 pages), 2013.
\newblock Addendum in \emph{J. Knot Theory Ramif.} 23(8):1491001 (2 pages),
  2014.

\bibitem{k-etscn-1899}
Arthur~Edwin Kennelly.
\newblock Equivalence of triangles and three-pointed stars in conducting
  networks.
\newblock \emph{Electrical World and Engineer} 34(12):413--414, 1899.

\bibitem{k-dg-97}
Mikhail Khovanov.
\newblock Doodle groups.
\newblock \emph{Trans. Amer. Math. Soc.} 349(6):2297--2315, 1997.

\bibitem{l-pubrm-15}
Marc Lackenby.
\newblock A polynomial upper bound on {Reidemeister} moves.
\newblock \emph{Ann. Math.} 182(2):491--564, 2015.
\newblock arXiv:\href{http://arxiv.org/abs/1302.0180}{1302.0180}.

\bibitem{l-wtpn-63}
Alfred Lehman.
\newblock Wye-delta transformations in probabilistic network.
\newblock \emph{J. Soc. Indust. Appl. Math.} 11:773--805, 1963.

\bibitem{m-fppa-90}
Ketan Mulmuley.
\newblock A fast planar partition algorithm, {I}.
\newblock \emph{J. Symbolic Comput.} 10(3--4):253--280, 1990.

\bibitem{nt-aafts-96}
Hiroyuki Nakahara and Hiromitsu Takahashi.
\newblock An algorithm for the solution of a linear system by {$\Delta$-Y}
  transformations.
\newblock \emph{IEICE TRANSACTIONS on Fundamentals of Electronics,
  Communications and Computer Sciences} E79-A(7):1079--1088, 1996.
\newblock Special Section on Multi-dimensional Mobile Information Network.

\bibitem{nw-kg-00}
Steven~D. Noble and Dominic J.~A. Welsh.
\newblock Knot graphs.
\newblock \emph{J. Graph Theory} 34(1):100--111, 2000.

\bibitem{n-cpsc-09}
Tahl Nowik.
\newblock Complexity of planar and spherical curves.
\newblock \emph{Duke J. Math.} 148(1):107--118, 2009.

\bibitem{p-cails-02}
Jane~M. Paterson.
\newblock A combinatorial algorithm for immersed loops in surfaces.
\newblock \emph{Topology Appl.} 123:205--234, 2002.

\bibitem{p-icfgd-98}
Michael Polyak.
\newblock Invariants of curves and fronts via {Gauss} diagrams.
\newblock \emph{Topology} 37(5):989--1009, 1998.

\bibitem{r-ebk-27}
Kurt Reidemeister.
\newblock Elementare {Begründung} der {Knotentheorie}.
\newblock \emph{Abh. Math. Sem. Hamburg} 5:24--32, 1927.

\bibitem{rs-gm10-91}
Neil Robertson and Paul~D. Seymour.
\newblock Graph minors. {X.} {Obstructions} to tree-decomposition.
\newblock \emph{J. Comb. Theory Ser. B} 52(2):153--190, 1991.

\bibitem{rst-qepg-94}
Neil Robertson, Paul~D. Seymour, and Robin Thomas.
\newblock Quickly excluding a planar graph.
\newblock \emph{J. Comb. Theory Ser. B} 62(2):232--348, 1994.

\bibitem{r-md-1904}
Alexander Russell.
\newblock The method of duality.
\newblock \emph{A Treatise on the Theory of Alternating Currents}, chapter
  XVII, 380--399, 1904. Cambridge Univ. Press.

\bibitem{s-iifpd-01}
Xiaohuan Song.
\newblock Implementation issues for {Feo} and {Provan}'s delta-wye-delta
  reduction algorithm.
\newblock {M.Sc.} {Thesis}, University of Victoria, 2001.

\bibitem{s-pr-1916}
Ernst Steinitz.
\newblock Polyeder und {Raumeinteilungen}.
\newblock \emph{Enzyklopädie der mathematischen Wissenschaften mit Einschluss
  ihrer Anwendungen} III.AB(12):1--139, 1916.

\bibitem{sr-vtp-34}
Ernst Steinitz and Hans Rademacher.
\newblock \emph{Vorlesungen über die Theorie der Polyeder: unter Einschluß
  der Elemente der Topologie}.
\newblock Grundlehren der mathematischen Wissenschaften~41. Springer-Verlag,
  1934.
\newblock Reprinted 1976.

\bibitem{s-mncpe-84}
James~A. Storer.
\newblock On minimal node-cost planar embeddings.
\newblock \emph{Networks} 14(2):181--212, 1984.

\bibitem{t-drpg-89}
Klaus Truemper.
\newblock On the delta-wye reduction for planar graphs.
\newblock \emph{J. Graph Theory} 13(2):141--148, 1989.

\bibitem{t-md-92}
Klaus Truemper.
\newblock \emph{Matroid Decomposition}.
\newblock Academic Press, 1992.

\bibitem{v-ucvc-81}
Leslie~S. Valiant.
\newblock Universality considerations in {VLSI} circuits.
\newblock \emph{IEEE Trans. Comput.} C-30(2):135--140, 1981.

\bibitem{v-kfdp-89}
Gert Vegter.
\newblock Kink-free deformation of polygons.
\newblock \emph{Proceedings of the 5th Annual Symposium on Computational
  Geometry}, 61--68, 1989.

\end{thebibliography}

\end{document}